\documentclass[trsc,dblanonrev]{informs4}

\OneAndAHalfSpacedXI

\usepackage{amsmath,amssymb,amsfonts,mathtools}
\usepackage{comment}
\usepackage{tabularx} 

\RequirePackage{endnotes}
\usepackage{anyfontsize}
\usepackage{enumerate}

\usepackage{natbib}
 \bibpunct[, ]{(}{)}{,}{a}{}{,}%
 \def\bibfont{\small}%

\usepackage{rotating}
\usepackage{fancyvrb}

\TheoremsNumberedThrough     
\ECRepeatTheorems
\JOURNAL{XXX}
\newtheorem{corollary}{Corollary}[proposition]

\EquationsNumberedThrough    

\begin{document}




\TITLE{\large How bad is time variability for users in mobility services? An economic framework under expected- and non-expected utility}

\ARTICLEAUTHORS{%
\AUTHOR{David Z.W. Wang \textsuperscript{a}, Zhaoqi Zang \textsuperscript{a}*, Xiangdong Xu \textsuperscript{b} and Shaojun Liu \textsuperscript{a}}
\AFF{\textsuperscript{a} School of Civil and Environmental Engineering, Nanyang Technological University, 50 Nanyang Avenue, Singapore 639798, Singapore, \\
\textsuperscript{b} College of Transportation, Tongji University, Shanghai, China,\\
\EMAIL{* Corresponding author. zhaoqi.zang@ntu.edu.sg}; wangzhiwei@ntu.edu.sg;xiangdongxu@tongji.edu.cn;lius0075@e.ntu.edu.sg}


}

\ABSTRACT{\normalsize Time variability is a pervasive feature of mobility services and a major source of welfare loss for users. Although a substantial literature has quantified the cost of time variability (COTV), it remains theoretically unclear \textit{how bad time variability can be for users in the worst case}. In the absence of such a worst-case benchmark, quantified variability costs lack a principled reference for assessing whether they are economically meaningful or negligible. Meanwhile, such a benchmark is critical for strategic prioritization in transport appraisal, service design, and pricing—particularly in early-stage decision making, where detailed valuation is often infeasible. To fill this gap, this paper first develops an economic framework to quantify the cost of time (COT) and the COTV under expected utility model, and to establish theoretical upper bounds on the ratio of COTV to COT. For users with quadratic utility, we show that $\text{COTV}/\text{COT} \le \tfrac{1}{2}CV^{2}$, where $CV$ is the coefficient of variation of service time. When the service process follows a Poisson process, a commonly adopted assumption, this bound simplifies to $\text{COTV}/\text{COT} \le \tfrac{1}{2}$, implying that time variability can inflate a user's total time cost of taking a mobility service by up to 50\%. In more general settings, the ratio depends on three interpretable factors: $CV$ and users' second- and third-order risk preferences, captured by relative risk aversion (RRA) and relative prudence (RP). We identify benchmark values of RRA and RP that characterize users' preferences over mean-, variance-, and skewness-related reductions. Finally, we extend our framework to non-EU models, including dual theory and rank-dependent utility, demonstrating that analogous structural decompositions arise under DT and RDU, although the relevant variability measures differ across preference models.
}

\maketitle

\section{Introduction}
The economic costs of time variability—the inconsistency or unpredictability of transport or mobility service times—are profound \textcolor{blue}{\citep{winston2013performance}}. For U.S. motorists, truckers, and shippers, increased delays due to time variability resulted in a \$45 billion welfare loss \textcolor{blue}{\citep{langer2008toward}}, while in U.S. domestic air travel, arrival delay variability doubled users' expected costs \textcolor{blue}{\citep{koster2016user}}. Moreover, time variability significantly affects users' satisfaction and loyalty across various mobility services and systems, including public transport, online delivery or ridesourcing, and tolling or charging services \textcolor{blue}{\citep[e.g.,][]{yin2009dynamic,jung2014stochastic,ansari2021waiting,xu2021longer,yu2022delay}}. Recognizing the inevitability of time variability, users are willing to pay to reduce it (see \textcolor{blue}{\cite{carrion2012value}}), prompting service providers to improve reliability and regulators to require the inclusion of time variability costs in project appraisals \textcolor{blue}{\citep{NZTA2010economic,de2015including}}. Consequently, quantifying the cost of time variability (COTV) is of interest to all stakeholders, including users, service providers, regulators, and the society. Accordingly, a wide range of empirical and theoretical methods have been therefore developed to quantify COTV  \textcolor{blue}{\citep[see reviews e.g.,][]{li2010willingness,carrion2012value,zang2022ttrb}}. Together, these studies have substantially advanced our understanding of how variability affects user' behavior.

However, quantifying COTV alone, often referred to as the value of reliability (VOR) in the literature, only addresses the question of \textit{How much does it cost users due to time variability}. In other words, the VOR studies do not answer another fundamental question: \textit{How bad can time variability be for users in the worst case}? Without a theoretical benchmark, the estimated COTVs lack an interpretable scale: while they indicate how much users are willing to pay for reliability, they do not reveal whether an estimated variability cost is economically meaningful or negligible. This limitation is especially consequential in early-stage decision-making, where reliability improvement investments (such as increasing road capacity and building dedicated lanes or a new road) must be screened and prioritized before a detailed evaluation is feasible.

This paper fills this gap by developing a framework to evaluate how large the economic inefficiency induced by time variability can be, relative to the cost users already incur from mean time. Specifically, we first model how users respond to time variability using risk premium approach. We then adopt the expected utility (EU) framework to estimate the cost of time variability (COTV) relative to the cost of time (COT), where the latter represents users’ cost under an otherwise identical deterministic service. The ratio of the COTV to COT is used as a normalized measure of the severity of time variability. Furthermore, we study marginal benefits by comparing the returns to reductions in time variability and mean service time via the ratio of the value of time variability to the value of time, capturing how users value a unit decrease in variability relative to a unit decrease in mean service time.

We show that under a commonly adopted setting for mobility services, where users' utility function is quadratic, and its service process follows a Poisson process, the COTV is at most $1/2$ of the COT. Consequently, the total cost for users under time variability is bounded by 1.5 times that of an otherwise identical deterministic service. This simple yet powerful bound provides the first theoretical “worst-case” benchmark for COTV. More generally, if we relax the Poisson process assumption and retain only the quadratic utility assumption, the ratio satisfies $COTV/COT \le 1/2 CV^2$ where $CV$ is the coefficient of variation of service time. This bound is appealing as it depends linearly only on readily observable characteristics of time variability (i.e., $CV^2$) for arbitrary risk preference parameters. As such, it is particularly useful for the early-stage decision-making related to strategic prioritization in transport appraisal, service design, and pricing. 

 
In an even more general setting where users' utility functions are only assumed to be differentiable, we show that the ratio of COTV to COT depends on three interpretable factors: the degree of time variability (captured by CV), users’ second-order risk preference (captured by relative risk aversion, RRA), and third-order risk preference (captured by relative prudence, RP). We identify benchmark values of RRA and RP that classify users according to their preferences for reducing the mean, variance, and skewness of random mobility service times. Intuitively, reduction in mean time makes trips faster, reduction in variance decreases spread of travel times around mean, and reduction in skewness protects users against rare but severe delays. This preference-based classification is of clear practical relevance, as identifying user types is essential for effectively managing time variability in service systems \textcolor{blue}{\citep{yu2018managing}}. 

Finally, recognizing the well-known limitations of expected utility (EU) models, such as common-ratio and framing effects, we extend our analysis to non–EU frameworks, including dual theory (DT; \textcolor{blue}{\cite{yaari1987dual}}) and rank-dependent utility (RDU; \textcolor{blue}{\cite{quiggin1982theory}}). Although these frameworks differ in their modeling primitives, the structural form of the COTV–COT ratio remains unchanged, continuing to depend on the degree of time variability and users’ risk preferences. The distinction lies in representation: under EU, variability and preferences are characterized in the payoff (i.e., mobility service time) plane via space through moments of time and derivatives of the utility function, whereas under non-EU models they are captured in the probability plane through dual moments and probability-weighting derivatives.

In summary, the \textbf{main contribution} of this paper is to provide a theoretical answer to how bad time variability can be for users. We show that the economic impact of time variability is inherently bounded and depends on both the degree of variability and users’ risk preferences, thus placing variability costs within a clear and interpretable benchmark relative to time costs. Taken together, these results complement existing valuation methods by offering a data-light benchmark for early-stage decision-making, as well as practical guidance for transport appraisal, reliability-oriented pricing, and the design of mobility services under uncertainty.


The remainder of this paper is organized as follows. Section 2 reviews time variability and its quantification. Section 3 models users’ response to time variability under EU framework. Section 4 quantifies the COTV and defines the ratios to assess how bad time variability can be. Section 5 gives theoretical results and discussions. Section 6 extends the analysis to non–EU frameworks. Section 7 concludes the paper.

\section{Literature review}
This section first reviews extensive evidence on the importance and consequences of time variability, and then introduces methods for evaluating the cost of time variability.

\subsection{Time variability in mobility service systems}
There are many forms of time variability in mobility service systems, most notably variability in travel times, waiting times, and queuing or service times.

Travel time variability (TTV) has long been recognized as an important topic since \textcolor{blue}{\cite{herman1974trip}}. Moreover, the rare events included in TTV may cause much more serious delays as expected by travelers, which is at most five times than that of common condition \textcolor{blue}{\citep{van2008travel}}. 
Reducing the TTV can therefore generate substantial benefits, often comparable to reducing travel time, as reported by many studies \textcolor{blue}{\citep{franklin2009travel,carrion2012value,devarasetty2012value,zang2024relia}}. This underscores the need to quantify the cost of TTV \textcolor{blue}{\citep{NZTA2010economic,de2015including,zang2024value}}. Waiting time variability plays a critical role in shaping the attractiveness and perceived service quality of mobility service in transport systems, including public transit \textcolor{blue}{\citep{cats2017modeling,ansari2021waiting}}, airports \textcolor{blue}{\citep{gkritza2006airport}}, and online delivery or ride-sourcing platforms \textcolor{blue}{\citep{xu2021longer,yu2022delay,cui2023sooner}}. For instance, \textcolor{blue}{\cite{fan2016waiting}} reported that passengers’ perceived waiting time was on average 1.21 times their actual waiting time. Consequently, waiting time is often regarded as the dominant component of generalized cost in public transport \textcolor{blue}{\citep{ansari2021waiting}}, and thus serves as a key performance indicator of service attractiveness \textcolor{blue}{\citep{esfeh2022waiting}}. Queuing and service time variability typically emerges when demand persistently exceeds service capacity of transport systems. An example is severe congestion at corridor bottlenecks—such as toll stations—during peak hours \textcolor{blue}{\citep{moghaddam2017smart}}, where long and uncertain queues significantly reduce operational efficiency. Beyond road networks, electric vehicle users often experience substantial and unpredictable queuing delays at urban public charging stations \textcolor{blue}{\citep{jung2014stochastic}}, exacerbating their range anxiety \textcolor{blue}{\citep{noel2019fear}}. 

The widely existing time variability and its significant effects highlight the necessity of understanding users’ behavioral responses to variability and further evaluating its cost for users or systems.

\subsection{Modeling the cost of time variability }

A notable theoretical achievement for modeling the time variability of transportation systems is the bottleneck model developed by \textcolor{blue}{\cite{vickrey1969congestion}}. \textcolor{blue}{\cite{small1982scheduling}} introduced the schedule cost into the bottleneck model and proposed the expected utility model, which assumes that travelers aims to maximize their trip utility in their trip scheduling and their disutility comes from not arriving at the preferred arrival time on time. Specifically, under this schedule delay framework, the concept of the monetary value of travel time reliability, namely the VOR, and reliability ratio approach have been developed to quantitatively assess the cost of travel time variability (TTV) \textcolor{blue}{\citep{small1982scheduling, noland1995travel, bates2001valuation,fosgerau2010value,jenelius2012vor,li2012embedding,engelson2016cost,zang2024value}}. This line of research is rooted in schedule-based travel behavior and variability causes cost because it generates earliness and lateness penalties, so we often refer to it as the \textit{scheduling delay} approach.

Another line of research on modeling the cost of time variability (COTV) is the \textit{Bernoulli approach}. The Bernoulli approach is also grounded in expected utility theory, but it explicitly attributes disutility to variability itself. Under this framework, travelers dislike time variability because they are risk-averse toward uncertain travel times, implying that greater risk aversion leads to higher cost of variability. Specifically, with the theoretical foundation of the so-called safety margin provided by  \textcolor{blue}{\cite{gaver1968headstart}} and \textcolor{blue}{\cite{knight1974approach}}, \textcolor{blue}{\cite{jackson1982empirical}} successfully used mean-variance preferences over random travel times to explain the choice among risky trips. In the context of a general travel choice model, \textcolor{blue}{\cite{senna1994influence}} further made notable advances in deriving measures of both the value of time and the value of travel time variability. Recently, \textcolor{blue}{\cite{beaud2016impact}} and \textcolor{blue}{\cite{zang2024relia}} used the Bernoulli approach to model the cost of travel time variability and explored the impact of travel time variability and travelers’ risk attitudes on it.

Since our goal is to quantify how much users suffer from time variability, we focus on the cost generated by variability itself. Accordingly, we adopt the Bernoulli approach to model the COTV. To the best of our knowledge, the current literature in quantifying the COTV has focused mainly on deriving its formulations and examining the factors influencing it \textcolor{blue}{\citep[See reviews e.g.,][]{carrion2012value, zang2022ttrb}}. However, no studies have ever established an upper bound on the ratio of COTV to COT. As a result, how severe time variability can be for users remains an open question. The studies most closely related to our work are those on the VOR and the reliability ratio, from which our focus differs in several respects. First, the reliability ratio captures the \textit{marginal} value of reducing \textit{travel time} variability per unit of travel time, whereas we study the \textit{total} cost attributable to time variability. Second, recent advances in the VOR mainly express it as a function of users’ preferences and a variability measure, and then estimate it empirically using stated- or revealed-preference data; see the reviews in \textcolor{blue}{\cite{li2010willingness}} and \textcolor{blue}{\cite{carrion2012value}}. Such empirical quantification, however, lacks a theoretical benchmark for assessing whether a quantified variability cost is economically large or negligible. In other words, while VOR indicates users’ willingness to pay for reliability, it does not reveal how severe time variability can be relative to mean time in \textit{worst-case scenarios}. In contrast, a theoretical benchmark, specifically an upper bound on the ratio of COTV to COT, is particularly valuable for early-stage decision-making. For example, it benefits screening and prioritizing reliability-improvement options (e.g., capacity expansions or constructing a new road) before detailed project-level appraisal is feasible. Finally, most current VOR studies are based on the EU framework, which has well-known limitations (e.g., common-ratio and framing effects), whereas we extend our analysis to non-EU frameworks.

\section{A theoretical framework for modeling users' responses to time variability} \label{sec2}
This section introduces mobility service users' utility function of time and their risk attitudes, and models their responses to time variability, thereby providing a theoretical foundation for the subsequent analysis. 

For a transport or mobility service, we use lowercase $t \in \mathbb{R}_{+}$ to represent its service time\endnote{Note that $t$ can represent various realistic duration time involved in the mobility services; its specific terminology depends on specific contexts, e.g., travel time, waiting time, queuing time, and service time.} and uppercase $T$ to represent a constant value of $t$. Let $\mu$, $\sigma$, $f(t)$, and $F(t)$ denote its mean, standard deviation, probability density function (PDF), and cumulative distribution function (CDF), respectively. Throughout this paper, we refer to the single $(T)$ as an \textit{instance without time variability} with utility $u(T)$; the triple $(t, \mu,\sigma)$ referred to as an \textit{instance with time variability} and and its expected utility is denoted as $\mathbb{E}[u(t, \mu,\sigma)]$. For notational convenience, we use $(t)$ to denote the instance $(t, \mu,\sigma)$ when no ambiguity arises.

\subsection{Utility function of time and risk attitudes} \label{sec2.1}
In the context of economic risks, individuals generally prefer larger amounts of wealth $s \in  \mathbb{R}_{+}$ because it typically yields positive utility to them. Therefore, the utility function of wealth, denoted by $u(s)$, is assumed to be a non-decreasing function, i.e.,  $u'(s) \geq 0$. In contrast, as transport demand is derived from the need to move people or goods, time spent for mobility services typically generates negative utilities; users therefore dislike long service times. For example, passengers generally dislike longer waiting times for buses. As a result, the utility function of time is a non-increasing function, i.e.,  $u'(t) \leq 0$ for $t \in \mathbb{R}_{+}$. 

Building seminar works on risk attitudes \textcolor{blue}{\citep{pratt1964risk,arrow1965aspects,kimball1990precautionary}}, Definition \ref{RARPRiskAttitudes} formalizes the second-order risk preference (i.e., Pratt-Arrow risk aversion) and the third-order risk preference (i.e., prudence) for mobility service users with differentiable utility function. Interested readers can further refer to \textcolor{blue}{\cite{gollier2001economics}} and \textcolor{blue}{\cite{eeckhoudt2006putting}} for understanding risk aversion and prudence. 

\vspace{0.5em}
\begin{definition}[Risk attitudes]\label{RARPRiskAttitudes}
Given a mobility service user with a differentiable utility function of service time, the user is (i) \textit{risk-averse} if and only if the utility function is concave and (ii) \textit{prudent} if and only if the marginal utility function is convex.
\end{definition}
\vspace{0.5em}

Definition \ref{RARPRiskAttitudes} implies that, mathematically, risk aversion is equivalent to $u\left( \mathbb{E}\left[ t \right] \right) \ge \mathbb{E}\left[ u\left( t \right) \right]$, which in turn implies $u''(t) \le 0$. For this reason, risk aversion is referred to as a second-order risk preference. Prudence is equivalently characterized by $u^{\prime}\left( \mathbb{E}\left[ t \right] \right) \le \mathbb{E}\left[ u^{\prime}\left( t \right) \right]$, i.e., $u^{\prime\prime\prime}(t) \ge 0$, and therefore is termed as a third-order risk preference\endnote{This paper follows the original meaning of prudence introduced in economic risks \textcolor{blue}{\citep{eeckhoudt2006putting}}: agents raise savings for uncertainty in future incomes. That is, prudent users of mobility services prepare to pay additional time for variability in their service times. This is different with \textcolor{blue}{\cite{beaud2016impact}} that defined prudence by $u^{\prime\prime\prime}(t) \le 0$ under context of travel time variability from the perspective of disliking upside variability.}. In the context of time variability, risk aversion reflects users’ dislike of variability and their desire to avoid it whenever possible, whereas prudence captures users’ propensity to prepare for and protect themselves against rare but severe unfavorable realizations of random service time.


\subsection{Modeling users' responses to time variability} \label{sec2.2}
For a given mobility service, consider two instances: one determinstic benchmark without time variability $(T_0)$ and one stochastic instance with time variability $(t, T_0,\mathbb{E}[\tilde{\varepsilon}^2])$ by adding a zero-mean time variability $\tilde{\varepsilon}$ into  $(T_0)$. Figure \ref{figureServiceWithRiskEU} shows these two instances by using a discrete representation of the variability $\tilde{\varepsilon}$ for illustrative clarity; the analysis readily extends to continuous representations as $n \rightarrow \infty$. Throughout the paper, values shown along the lines (at the right end) in all figures represent probabilities (outcomes).
\begin{figure}[ht]
    \centering
    \vspace{-0.8em}
    \includegraphics[width=0.90\textwidth]{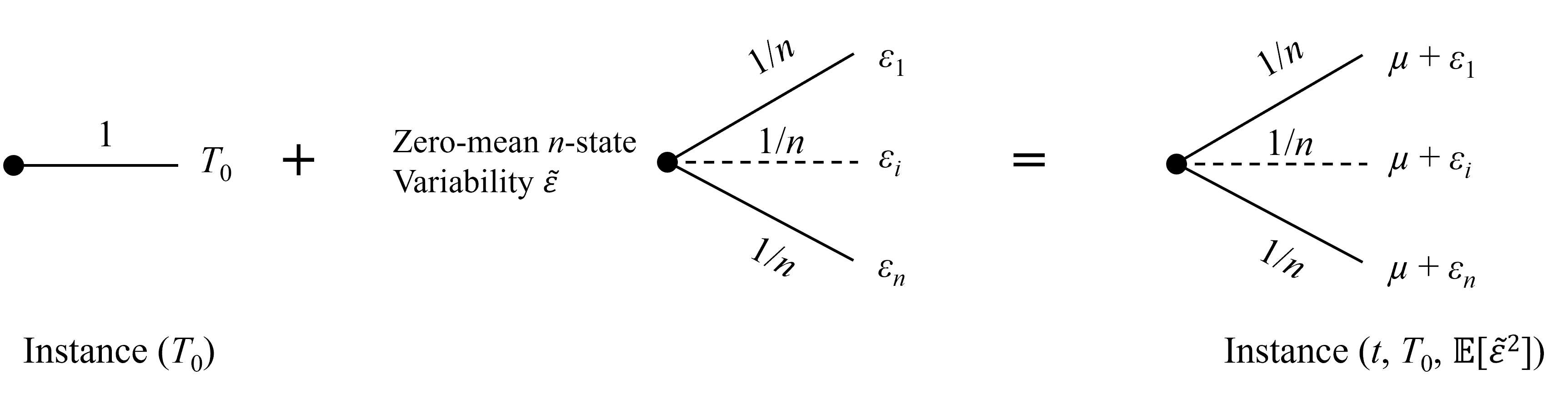}
    \vspace{-0.8em}
    \caption{Service instances without time variability $(T_0)$ and with time variability $(t, T_0,\mathbb{E}[\tilde{\varepsilon}^2])$.}
    \label{figureServiceWithRiskEU}
\end{figure}

For users who prefer certainty to variability, a common behavioral practice is to reserve additional time $\delta$ beyond $T_0$ to eliminate the impact of time variability. This practice corresponds to safety margin hypothesis of how users manage time variability, which has been widely recognized in both empirical and theoretical studies  \textcolor{blue}{\citep[e.g.,][]{gaver1968headstart,lo2006degradable, zang2024relia}}. Given that the utility function of time is non-increasing, the maximum $\delta$ should satisfy $u(T_0 + \delta_{\max}) = \mathbb{E}[u(t)]$. We can find that $\delta_{\max}$ identifies a utility-equivalent deterministic instance $(T_0 + \delta_{\max})$ that is indifferent to the instance with time variability $(t, T_0,\mathbb{E}[\tilde{\varepsilon}^2])$. Besides, $T_0$ and $\mathbb{E}[\tilde{\varepsilon}^2]$ are the mean and variance of the random variable $t$, namely $\mu =\mathbb{E}[t] = T_0$ and $\sigma^2 =\mathbf{Var}[T_0 +\tilde{\varepsilon}] = \mathbf{Var}[\tilde{\varepsilon}] = \mathbb{E}[\tilde{\varepsilon}^2]$. Therefore, we formally define such $\delta_{\max}$ as variability premium in Definition \ref{variability premium}, as our focus here is on mobility service time variability . This definition is consistent with the classical notion of risk premium in economic theory with focus on risks \textcolor{blue}{\citep{pratt1964risk,arrow1965aspects}}. Specifically, for a given mobility service, users are \textit{indifferent} between the instance with time variability $(t)$ and the newly identified equivalent instance without variability $(\mu + \pi)$, as shown by Figure \ref{figureIndifferentRisk}. 
\begin{figure}[ht]
    \centering
    \vspace{-0.8em}
    \includegraphics[width=0.65\textwidth]{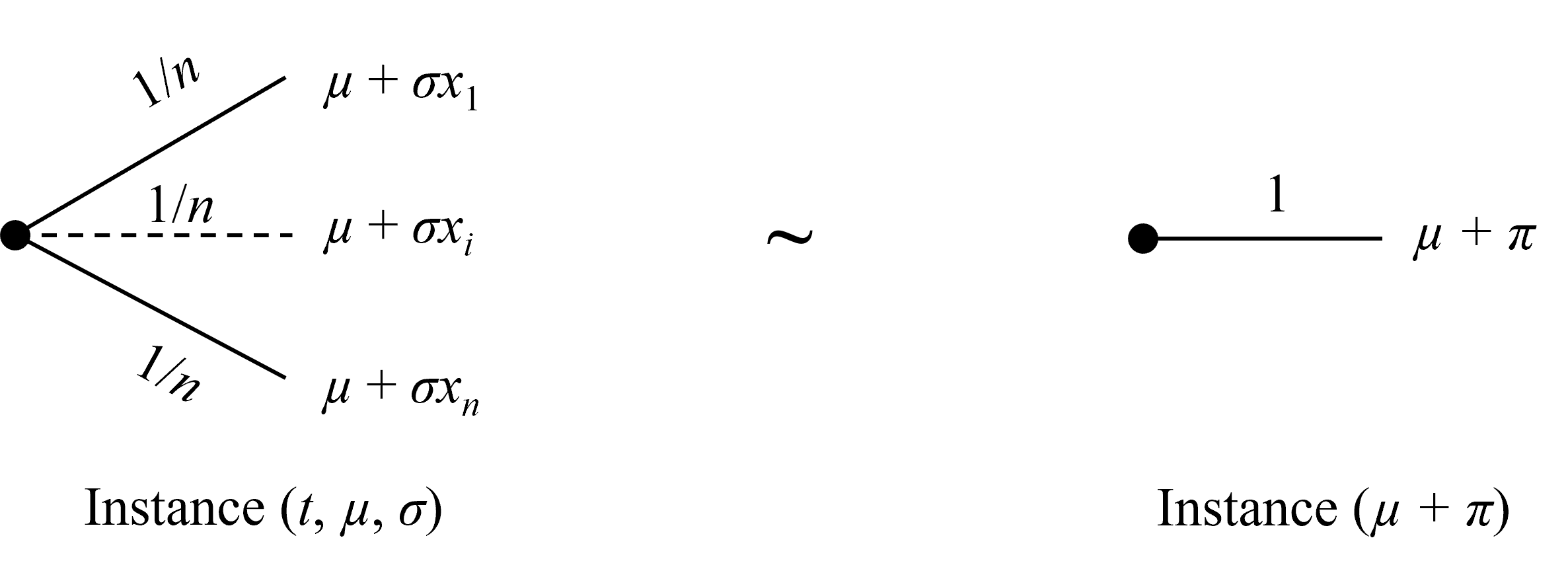}
    \vspace{-0.8em}
    \caption{Indifferent between two service instances and variability premium}
    \label{figureIndifferentRisk}
\end{figure} 

\vspace{-0.5em}
\begin{definition}[Variability premium]\label{variability premium}
    The variability premium $\pi$ measures the amount of additional time a user is willing to pay beyond the expected time to eliminate the impact of time variability in a service, which mathematically satisfies  $\mathbb{E}[u\left( t \right)] =u\left( \mu +\pi \right)$. 
\end{definition}
\vspace{0.5em}

Let $ t=\mu +\sigma x$, where $x$ is a standardized time variable with zero mean and $\mathbb{E}\left[ x \right] =\mathbb{E}\left[ {(t-\mu)}/{\sigma} \right] =0\,\,\text{and\,\,}\mathbb{E}\left[ x^2 \right] =\mathbb{E}\left[ {\left( t-\mu \right) ^2}/{\sigma ^2} \right] = 1$. Following the terminology of \textcolor{blue}{\cite{pratt1964risk}} in economics, $x$ can be viewed as the actuarially neutral variability, while the parameter $\sigma$ serves as a scale parameter capturing the magnitude of this variability. Proposition \ref{RP_local approxiamtion} formalizes the concept of the variability premium in the context of time variability, building on the seminal works of \textcolor{blue}{\cite{pratt1964risk}} and \textcolor{blue}{\cite{ arrow1965aspects}}, as well as subsequent developments in the transportation literature \textcolor{blue}{\citep{batley2007marginal,beaud2016impact,zang2024relia}}.



\begin{proposition}\label{RP_local approxiamtion}
    Consider a service whose service time $t$ is a random variable with mean $\mu$ and variance $\sigma^2$. For a user with utility function of $u(t)$, the variability premium is given by
    \begin{equation} \label{RP_local approxiamtion_EQ}
        \pi \simeq \frac{\sigma^2}{2} \frac{u^{\prime\prime}(\mu)}{u^\prime(\mu)} \notag,
    \end{equation}
\end{proposition}
\begin{proof}
    {Proof.} 
    Since $\pi$ is developed to quantify the cost of time variability and $\sigma$ is a measure of the variability, it is natural to posit a functional relationship between $\pi$ and $\sigma$, i.e., $\pi = g_1(\sigma)$. Then, 
    \begin{equation}
        \mathbb{E}[u\left( \mu +\sigma x \right)] = u\left( \mu + g_1\left( \sigma \right) \right) \notag \label{EU_U_g_VP}.
    \end{equation}
    Differentiating it with respect to $\sigma$ once and twice yields, respectively 
    \begin{align}
        & \mathbb{E}\left[ xu'\left( \mu +\sigma x \right) \right] =u'\left( \mu +g_1\left( \sigma \right) \right) g_1'\left( \sigma \right)  \label{EU_U_g_1} \\
        & \mathbb{E}\left[ x^2u''\left( \mu +\sigma x \right) \right] =u''\left( \mu +g_1\left( \sigma \right) \right) \left( g_1'\left( \sigma \right) \right) ^2+g_1''\left( \sigma \right) u'\left( \mu +g_1\left( \sigma \right) \right) \label{EU_U_g_2}
    \end{align}
    Since the absence of time variability implies zero variability premium, so $\pi = g_1(0) = 0$ when $\sigma = 0$,. Moreover, recalling that $\mathbb{E}[x]=0$, it has $g_1'\left( 0 \right) = 0$ following from Eq. \eqref{EU_U_g_1}; and we further drive $g_1''\left( 0 \right) =\mathbb{E}\left[ x^2r\left( \mu \right) \right] ={u^{\prime\prime}(\mu)}/{u^\prime(\mu)}$ by evaluating Eq. \eqref{EU_U_g_2} at $\sigma = 0$. Using second-order Taylor expansion, we obtain 
    \begin{equation}
        \pi \simeq g_1(\sigma) = g_1(0) + g_1'\left( 0 \right) \sigma + \frac{g_1''\left( 0 \right)}{2} \sigma^2 = \frac{\sigma ^2}{2}\frac{u^{\prime\prime}(\mu)}{u^\prime(\mu)} \notag \label{RP}
    \end{equation}
    This completes the proof. $\square$
\end{proof}

It is worth to point out the term ${u^{\prime\prime}(\mu)}/{u^\prime(\mu)}$ in Proposition \ref{RP_local approxiamtion} corresponds to the coefficient of absolute risk aversion (ARA) at time $\mu$, which characterizes users' intensity of risk aversion. Another typical measure of the intensity of risk aversion is the coefficient of relative risk aversion (RRA). Let $A_2(t)$ and $R_2(t)$ denote ARA and RRA wherein the subscript 2 represents the 2nd-order risk preference, respectively, then\endnote{$A_2(t)$ and $R_2(t)$ are slightly different to their counterparts in the context of economic risks where $A_2(s) = -{ u''(s)}/{u'(s)}$, and $R_2(s) = -s{u''(s)}/{u'(s)}$. The negative sign in $A_2(s)$ and $R_2(s)$ is used to keep them positive as a risk averse customer has $u^{\prime\prime}(s) \le 0$ and $u^\prime(s) > 0$. But under time variability, a risk averse user has $u^{\prime\prime}(t) \le 0$ and $u^\prime(t) < 0$. So, there is no need of a negative sign.}
\begin{equation}
    A_2(t) = \frac{ u''(t)}{u'(t)}, \text{   } R_2(t) = t\frac{ u''(t)}{u'(t)}
\end{equation}

Therefore, Proposition \ref{RP_local approxiamtion} provides an \textit{insightful decomposition} of the variability premium, showing that  the amount of additional time users are willing to pay to eliminate time variability depends on the first two moments (i.e., expected value and variance) of times and the absolute risk aversion (ARA). For a given service within a fixed period, different users usually experience the same mean and variance; consequently, heterogeneity in the variability premium arises primarily from differences in users’ risk-aversion intensities. Besides, though not explicitly captured in Proposition \ref{RP_local approxiamtion}, prudence also plays an important role in shaping the variability premium as it reflects users’ propensity to allocate additional time as a precaution against time variability. The intensity of prudent preference is typically characterized by the coefficients of absolute prudence (AP) or relative prudence (RP) coefficient \textcolor{blue}{\citep{gollier2001economics,eeckhoudt2009values}}. 
Similarly, let $A_3(t)$ and $R_3(t)$ denote AP and RP where the subscript $3$ represents the 3rd-order risk preference, then
\begin{equation}
    A_3(t) = -\frac{u^{\prime\prime\prime}(t)}{u''(t)}, \text{   } R_3(t) = -t\frac{u^{\prime\prime\prime}(t)}{u''(t)},
\end{equation}
where the negative sign in $A_3(t)$ and $R_3(t)$ is introduced to ensure that these measures remain positive.

The above coefficients of (relative) risk aversion and prudence play important roles in quantifying users' costs associated with time variability in mobility services,  to be shown later in this paper. Up to this point, variability premium and its approximation have enabled us to quantify how users behaviorally respond to time variability and to quantify the costs that mobility service users incur due to time and its variability.

\section{Bounding the inefficiency of time variability relative to no variability} \label{sec:COTandCOTV_EU}
When using mobility services, users typically incur a service cost. Accordingly, the users' utility and expected utility from using a mobility service with random time $t \in [T_{\min}, T_{\max}]$ can be expressed by:
\begin{align} 
    & U(c, T) = -\varphi c + u (T) \label{Utility} \\
    & \mathbb{E}\left[ U\left( c,t \right) \right] =\int{U\left( c,t \right) f\left( t \right) dt}=-\varphi c+\int{u\left( t \right) f\left( t \right) dt} \label{EU_expression}
\end{align}
where $c$ denotes the mobility service cost (e.g., ticket fare or booking fee) and $\varphi$ is marginal utility of wealth. 
\subsection{The costs of time and time variability}
\subsubsection{The cost of time}
With reference to Eq. \eqref{Utility}, if $T$ is marginally increased from $T$ to $T +dT$, the resulting utility would be  $U(c, T + dT)$. Consequently, we can define the monetary value of time (VOT) at time $T$ as follows:
\begin{equation} \label{VOTatT}
    VOT\left( T \right) =-\frac{1}{\varphi}\lim_{dT\rightarrow 0}\frac{U\left( c,\ T+dT \right) -U\left( c,\ T \right)}{dT}=-u'\left( T \right) \frac{1}{\varphi}
\end{equation}


Based on the monetary VOT at time $T$ defined in Eq. \eqref{VOTatT}, we can further obtain the (average) VOT, which quantifies the monetary value users place on time savings per time unit \textcolor{blue}{\citep{li2010willingness}}. Namely,
\begin{equation} \label{EU_U_VOT}
    \mathbb{E}\left[ U\left(c,\,t \right) \right] =\mathbb{E}\left[ U\left(c+VOT\Delta t,\,t-\Delta t \right) \right] 
\end{equation}

Then, it is straightforward to compute the cost of time (COT) as 
\begin{equation} \label{EU_U_COT}
    COT = VOT \cdot \mu 
\end{equation}

We note that in some studies the terms $VOT(T)$,  $VOT$ and $COT$ are used interchangeably. In this paper, however, we deliberately distinguish among these three concepts. Although they are closely related, they represent different quantities, and their relationships are clarified in Remark \ref{DiffVOTCOT}.
\begin{remark}[Differences and relationship among $VOT(T)$, $VOT$, and $COT$] \label{DiffVOTCOT}
    $VOT(T)$ is the instantaneous monetary value of time evaluated at time $T$; $VOT$ represents the average monetary value that users place on time savings per unit of time; and $COT$ denotes the total monetary amount users are willing to pay for the entire duration of the mobility service. Therefore, the $VOT$ is obtained by aggregating $VOT(T)$ over time to be shown in Eq. \eqref{VOTusingVOTT}; and $COT$ is computed as the product of $VOT$ and the expected time spent for a mobility service, i.e., $COT = VOT \cdot \mu$. $\square$
\end{remark}

\subsubsection{The cost of time variability}
Based on Eqs. \eqref{Utility} and \eqref{EU_expression}, we can use the following equation \eqref{EU_U_COTV} to compute the additional monetary cost required to compensate for the utility loss induced by time variability, which is referred to as the cost of time variability (COTV). 
\begin{equation} \label{EU_U_COTV}
     \mathbb{E}\left[ U\left( c,\,\,t \right) \right] \,\,=\,\,U\left( c+COTV,\,\,\mu \right) 
\end{equation}

By solving Eq. \eqref{EU_U_COTV}, we can obtain the expression of COTV, summarized in Lemma \ref{EstimatingCOTV}.  The result is obtained by directly combining Eqs. \eqref{Utility} and \eqref{EU_expression} into Eq. \eqref{EU_U_COTV}. 

\begin{lemma}\label{EstimatingCOTV}
    The cost of time variability can be calculated by:
    \begin{equation} \label{COTV}
        COTV=\int{\frac{u\left( \mu \right) -u\left( t \right)}{\varphi}f\left( t \right) dt} \notag
    \end{equation}
\end{lemma}

So far, we have quantified the cost of using a mobility service without time variability (i.e., the $COT$) and the additional cost incurred by time variability (i.e., the $COTV$). These two measures provide a solid foundation for examining the impact of time variability on users, thereby enabling us address our central question:\textit{ how bad is time variability for users of mobility services}?

\subsection{The ratios}
\subsubsection{The ratio of the COTV to the COT}
To quantify the extent to which time variability amplifies the cost experienced by users, we define  $\rho \left(i, t, \mu \right)$ as the ratio of the COTV to the COT:
\begin{equation}\label{ratioExpre}
    \rho \left( i,t,\mu \right) =\frac{COTV}{COT}
\end{equation}

Based on Eqs. \eqref{VOTatT}, \eqref{EU_U_COT}, and \eqref{EU_U_COTV}, together with Proposition 1, we can derive the mathematical expression for the ratio $\rho \left(i, t, \mu \right)$. The result is summarized in  the following Proposition \ref{rho_local approxiamtion}.

%

\begin{proposition}\label{rho_local approxiamtion}
    For a mobility service with random time $t$ and a user with utility function $u(t)$, the ratio of the COTV to the COT can be approximated by
    \begin{equation} \label{rho_local approxiamtion_EQ}
        \rho \left( i,t,\mu \right) \simeq \frac{\pi}{\mu}\frac{1}{1-1/2\sigma ^2\frac{u^{\prime\prime\prime}\left( \mu \right)}{u'\left( \mu \right)}}
    \end{equation}
\end{proposition}

\begin{proof}
    {Proof.} We first derive the formulation of $COT$. Combining three equations of \eqref{Utility}, \eqref{VOTatT}, and \eqref{EU_U_VOT} gives
    \begin{equation}
        -\varphi VOT\Delta t=\int{u\left( t \right) f\left( t \right) dt}-\int{u\left( t-\Delta t \right) f\left( t \right) dt}
    \end{equation}    

    Then, by rearranging terms and based on $t = \mu + \sigma x$, we can have 
    \begin{equation}
        VOT=-\int{\frac{1}{\varphi}\frac{u\left( t \right) -u\left( t-\Delta t \right)}{\Delta t}f\left( t \right) dt}=\mathbb{E}\left[ VOT\left( t \right) \right] =\mathbb{E}\left[ -\frac{u'\left( \mu +\sigma x \right)}{\varphi} \right] 
    \end{equation}
    
    This implies that $VOT$ can be viewed as a function of $\sigma$; denote $VOT = g_2(\sigma)$. Then, the $l$-th derivative of the function $g_2(\sigma)$ is expressed by $g_2^{\left( l \right)}\left( \sigma \right) = -\mathbb{E}[x^lu^{\left( l+1 \right)}\left( \mu +\sigma x \right)/\varphi] $. Using second-order Taylor series expansion to $g_2(\sigma)$ around $\sigma = 0$ yields:
    \begin{equation} \label{VOTusingVOTT}
         VOT  \simeq g_2\left( 0 \right)  +\sum\limits_{l=1}^2{\frac{g_2^{\left( l \right)}\left( 0 \right)}{l!}\sigma ^l} = VOT(\mu) - 1/2 \sigma^2 \frac{u^{\prime\prime\prime}(\mu)}{\varphi}
    \end{equation}
    where the $VOT(\mu)$ is the defined VOT at time $\mu$. Finally, it is straightforward to compute the COT as 
    \begin{equation} \label{COT_formulation}
        COT\simeq VOT\left( \mu \right) \mu +\left( -\frac{1}{2}\frac{u^{'''}\left( \mu \right)}{\varphi}\sigma ^2 \right) \mu 
    \end{equation}
    
    Then, we derive the formulation of $COTV$. Since  $ t=\mu +\sigma x$, we have 
    $COTV=\mathbb{E}\left[ COTV\left( t \right) \right] =\mathbb{E}\left[ COTV\left( \mu +\sigma x \right) \right]$ where $COTV\left( t \right) =\frac{u\left( \mu \right) -u\left( t \right)}{\varphi}$. This representation allows us to view the $COTV$ as a function of $\sigma$, i.e., $ COTV=g_3\left( \sigma \right)$. Equivalently, we have $g_3\left( \sigma \right) =\mathbb{E}\left[ COTV\left( \mu +\sigma x \right) \right]$. Then, the $l$-th derivative of $g_3\left( \sigma \right)$ can be expressed by 
    \begin{equation}
          g_{3}^{\left( l \right)}\left( \sigma \right) =-\mathbb{E}\left[ x^lu^{\left( l \right)}\left( \mu +\sigma x \right) /\varphi \right], \text{ when } l \ge 1
    \end{equation}
    Using the second-order Taylor expansion, we can obtain an approximation of $COTV$ as follows.
    \begin{equation}\label{COTV_formulation}
        COTV=g_3\left( \sigma \right) \simeq g_3\left( 0 \right) +g_3'\left( 0 \right) \sigma +\frac{g_3''\left( 0 \right)}{2}\sigma ^2=\frac{1}{2}\frac{-u''\left( \mu \right)}{\varphi}\sigma ^2=\pi \times VOT\left( \mu \right) 
    \end{equation}
    According to Eq. \eqref{VOTatT}, we know $VOT\left( \mu \right) =-u'\left( \mu \right) \frac{1}{\varphi}$. Finally, combining Eqs. \eqref{COT_formulation} and \eqref{COTV_formulation} can obtain:
    \begin{equation}
        \rho \left( i,t,\mu \right) =\frac{\pi}{\mu}\frac{VOT\left( \mu \right)}{VOT\left( \mu \right) -1/2\sigma ^2\frac{u^{\prime\prime\prime}\left( \mu \right)}{\varphi}} = \frac{\pi}{\mu}\frac{1}{1-1/2\sigma ^2\frac{u^{\prime\prime\prime}\left( \mu \right)}{u'\left( \mu \right)}}\notag
    \end{equation}
    This completes the proof. $\square$
\end{proof}

Two observations immediately follow from Proposition \ref{rho_local approxiamtion} and its proof. 

First, Eq. \eqref{COTV_formulation} provides a straightforward way to monetize the time unit-based variability premium into monetary cost. Namely, we offer an alternative estimation method of the $COTV$ based on the proposed variability premium and the $VOT$ at time $\mu$ as given by Eq. \eqref{VOTatT}. This approach is particularly useful for incorporating the cost of time variability into cost–benefit analyses in transport project appraisal. For more details, please refer to \textcolor{blue}{\cite{zang2024relia}}. 

Second, the result of $\rho \left( i,t,\mu \right)$ depends critically on the derivatives of the utility function of time. As shown in Definition \ref{RARPRiskAttitudes}, users’ risk attitudes are characterized by these derivatives. In other words, risk attitudes play a central role in determining the magnitude of the ratio. Accordingly, when deriving the theoretical results in Section \ref{sec:theoreticalResult}, we first examine the most widely used utility function in pioneering empirical and theoretical transportation studies, i.e., the quadratic utility function, and then extend the analysis to a general setting in which the utility function of time is assumed to be differentiable.


\subsubsection{The ratio of the VOTV to the VOT}
We note that the ratio $\rho \left( i,t,\mu \right)$ captures the overall cost imposed on users by service time variability. However, in many practical settings, it is also important to consider the rate of return, in order to balance improvements in efficiency against efforts to eliminate variability. This consideration arises because demand for mobility services is derived demand; consequently, users typically face a limited time budget for mobility services. Motivated by this observation, we further examine the rate of return per unit of additional cost that users incur when hedging against time variability. This perspective allows us to compare the relative importance users attach to a marginal reduction in time variability versus a marginal reduction in mean time. We can notice that if we use variability premium as the valuation measure, Eq. \eqref{COTV_formulation} provides us a simple way to derive the value of time variability (VOTV), which quantifies the monetary value users place on a unit of reduction in time variability. That is, we have $VOTV\ =\ COTV/\pi =\ VOT(\mu) =\ -u'\left( \mu \right) \frac{1}{\varphi}$. Then, the ratio of the VOTV to VOT, which focuses on the rate of return, is given by\endnote{$\eta \left( i,t,\mu \right)$ is consistent with the so-called reliability ratio in existing valuation studies of travel time reliability \textcolor{blue}{\cite[See e.g.,][]{jackson1982empirical,fosgerau2010value,taylor2017fosgerau,zang2022ttrb}}.}:
\begin{equation} \label{RR}
     \eta \left( i,t,\mu \right) =\,\,\frac{VOTV}{VOT} = \frac{1}{1-1/2\sigma ^2\frac{u^{\prime\prime\prime}\left( \mu \right)}{u'\left( \mu \right)}}
\end{equation}


\section{Theoretical results} \label{sec:theoreticalResult}
This section first derives the ratios under quadratic utility functions and then extends the analysis to more general settings with arbitrary differentiable utility functions..

\subsection{Results under quadratic utilities}\label{Sec:SpecialQuadratic} 
In the literature, the utility function of time in mobility service systems is typically assumed to be quadratic, as it provides analytical tractability while capturing second-order risk preferences commonly assumed in transport economics \textcolor{blue}{\citep[See e.g.][]{jackson1982empirical,senna1994influence,yin2004new, beaud2016impact, zang2022ttrb}}. Under this assumption, users’ decision making depends on no more than second-order information of time variability, and Remark \ref{remarkQuadraRiskMoment} in next Section \ref{sec:Benchmark value of RRA} will provide theoretical analysis for this. This is reasonable in practice. Indeed, even the second moment of time distribution, i.e., variance, are already difficult for users to understand\endnote{Although variance is a commonly used measure in the mentioned VOR studies, it is challenging to communicate the extent of TTV clearly and simply to the public and thus there have been considerable efforts to design stated preference questionnaires to communicate the concept of TTV as straightforwardly as possible. See \textcolor{blue}{\cite{li2010willingness}} and \textcolor{blue}{\cite{carrion2012value}} for details.} \textcolor{blue}{\citep{FHWA2006,nevers2014guide}}, let alone the higher moment information (e.g., skewness) of time variability.

If users' utility function $u(t)$ is quadratic, we have $u^{\prime\prime\prime}\left( \mu \right) = 0$ and the ratio would reduces to
\begin{equation}
    \rho \left( i,t,\mu \right) = \frac{1}{2} R_2(\mu) CV^2 \label{generalQuadraRatio_EU}
\end{equation}
where $CV^2 = {\sigma^2}/{\mu^2}$.

Here, the ratio $\rho \left( i,t,\mu \right)$ depends linearly on the degree of users' second-order risk preference, i.e., $R_2(\mu)$, and the degree of time variability, i.e., $CV^2$. 
To derive an upper bound for the ratio, we first consider a commonly assumed service process in both theoretical and practical studies of mobility service systems, i.e., the Poisson process; and then extend the analysis to other arbitrary service processes. 


\subsubsection{An upper bound of the ratio $\rho$ for mobility service under Poisson processes} \label{PossionService}
Generally, mobility service operates as a stochastic process which usually involves queuing process of users due to its limited capacity. Consequently, service process in various mobility  systems is typically modeled as the Poisson process, under which the service duration or service time of users follows the exponential distribution. For example, exponential distribution is widely used to theoretically and empirically model travel time of private cars, taxi, and bus or transit \textcolor{blue}{\citep[e.g.,][]{jackson1982empirical,noland1995travel,yin2002optimal,benezech2013value,qian2021scaling}}, departure or arrival delay of bus or train or railway \textcolor{blue}{\citep[e.g.,][]{ cai2014optimal,harrod2019closed}}, and waiting or queuing time in offline mobility services such as charging or tolling services and online mobility services such as ridesourcing or delivery \textcolor{blue}{\citep[e.g.,][]{guo2007analysis,allon2011impact,sa2019dynamic,lin2023wait}}.  

In particular, Proposition \ref{quadraticRatioExponetialTime} states that, if the mobility service follows a Poisson process, users with quadratic function will incur at most $1/2$ times more additional cost due to the existence of time variability. In other words, the total cost of mobility service under time variability is at most $3/2$ times that under the setting with no time variability. Furthermore, such upper bound is tight as the bound is attained when the utility function is pure quadratic (i.e., $u(t) = at^2 +c$). Namely, $\rho\left( i,t,\mu \right) = \frac{1}{2}$. To the best of our knowledge, this simple yet insightful upper bound has not been previously established in the literature.
\begin{proposition}\label{quadraticRatioExponetialTime}
    An expected-utility (EU) user with quadratic utility function will pay at most 1/2 of the cost of time as additional cost due to  time variability in mobility systems whose service  follows the Poisson process. Formally,  $\rho\left( i,t,\mu \right) \le \frac{1}{2}$, with equality attained for a pure quadratic utility function. 
\end{proposition}
\begin{proof}
    {Proof.} If users have a quadratic utility function of time, then it can be written as $u(t) = at^2 + bt +c$ in which $u^{\prime\prime\prime}(t) =0$. According to Lemma \ref{RP_local approxiamtion} and noting that $A_2(\mu) = \frac{2a}{2a \mu + b}$, the ratio can be simplified as
    \begin{equation}
        \rho \left( i,t,\mu \right) =\frac{COTV}{COT}=\frac{\pi}{\mu}\frac{VOT\left( \mu \right)}{VOT\left( \mu \right) -1/2\sigma ^2\frac{u^{\prime\prime\prime}\left( \mu \right)}{\varphi}}=\frac{\pi}{\mu}\frac{VOT\left( \mu \right)}{VOT\left( \mu \right)}=\frac{\pi}{\mu}=\frac{\sigma ^2}{2\mu} \cdot \frac{2a}{2a\mu +b} \label{ratio_quadratic_1}
    \end{equation}
    As introduced in Section \ref{sec2.1}, for $t>0$, $u(t)$ should be a non-increasing function. This implies that $a <0$ and $-\frac{b}{2a} \leq 0$ or equivalently, $a <0$ and $b \leq 0$. Consequently, the maximum value of ratio $\rho \left( i,t,\mu \right)$ based on Eq. \eqref{ratio_quadratic_1} can be derived as follows
    \begin{equation}
        \rho \left( i,t,\mu \right) =\frac{\sigma ^2}{2\mu}\cdot \frac{2a}{2a\mu +b}\leq \frac{\sigma ^2}{2\mu}\cdot \frac{2a}{2a\mu}=\frac{1}{2}\frac{\sigma ^2}{\mu ^2}=\frac{1}{2}CV^2 \label{ratio_quadratic_2}
    \end{equation}    
    If the service of mobility systems follows a Poisson process, then the time of users involving in the service process will follow an exponential distribution. One key property of exponential distribution is that it has the same mean and standard deviation, implying $CV = 1$. Therefore, we have $\rho\left( i,t,\mu \right) \le  \frac{1}{2}$. For a pure quadratic utility function, i.e., $u(t) = at^2 +c$, the equality $\rho\left( i,t,\mu \right) = \frac{1}{2}$ holds. This completes the proof. $\square$
\end{proof}
\subsubsection{An upper bound of the ratio $\rho$ for general mobility service} \label{GeneralService}
If no specific assumption is imposed on the service process of the mobility systems, we can derive a general result for the ratio under an arbitrary random service process, which is summarized in Corollary \ref{quadraticRatioProp}. Its proof follows directly from that of Proposition \ref{quadraticRatioExponetialTime} and is therefore omitted.

\vspace{0.5em}
\begin{corollary}\label{quadraticRatioProp}
    For an EU user with quadratic utility function, $\rho\left( i,t,\mu \right) \le \frac{1}{2} CV^2$; particularly, for a pure quadratic utility function, equality holds, i.e., $\rho\left( i,t,\mu \right) = \frac{1}{2} CV^2$.  
\end{corollary} 
\vspace{0.5em}

Corollary \ref{quadraticRatioProp} indicates that, for a risk-averse user, the COTV is at most $1/2 CV^2$ times the COT. In other words, the upper bound of $\rho\left( i,t,\mu \right)$ is exactly $1/2 CV^2$. Furthermore, such upper bound is tight as the equation is attained when the utility function is purely quadratic (i.e., $u(t) = at^2 +c$). Section \ref{ImplicationOfQuadratic} will discuss the underlying implications of Corollary \ref{quadraticRatioProp}.

\subsubsection{Users value time variability reductions as much as time reductions} 
With respect to the benefits from time variability reduction and mean service time reduction, Corollary \ref{EqualRRas1} shows that users obtain the same benefit per unit of time saved as the amount that they pay per unit of time variability reduced. Its proof is straightforward: when the utility function $u(t)$ is quadratic, the ratio $\eta \left( i,t,\mu \right)$ defined in Eq. \eqref{RR} would be 1, since $u^{\prime\prime\prime}\left( \mu \right) = 0$.

\vspace{0.5em}
\begin{corollary}\label{EqualRRas1}
    For risk-averse EU users with quadratic utility function, the value of time variability equals the value of time. Namely, $\eta \left( i,t,\mu \right) = 1$.
\end{corollary}
\vspace{0.5em}

The result given by Corollary \ref{EqualRRas1} actually provides a theoretical support for a commonly stated claim in the literature \textcolor{blue}{\cite[see e.g.,][]{chen2023conservative, zang2024relia}} highlighting the importance of accounting for the COTV: \textit{users value time variability reductions as much as, if not more than, time savings}. In addition, it also lends theoretical justification to the prevalent practice in reliability-based network design problem, where the same assumed VOT coefficient is directly used to transfer the additional time for hedging against time variability to monetary costs (e.g., \textcolor{blue}{\cite{yan2013robust}}). For more details, please see a review paper of \textcolor{blue}{\cite{chen2011transport}}. Remark \ref{remarkQuadraRiskMoment} in \ref{sec:Benchmark value of RRA} elaborates on the underlying theoretical implications of adopting a quadratic utility specification for users in mobility service systems.


\subsubsection{Implications of special results of $1/2$ and $1/2CV^2$} \label{ImplicationOfQuadratic}
According to \textcolor{blue}{\cite{wardman2016values}}, the congestion multiplier in travel time valuation is the ratio of the VOT with congestion to the VOT without congestion. Many empirical studies report 1.5 as the central value of the congestion multiplier \textcolor{blue}{\citep{wardman2016values}} and this parameter plays a pivotal role in transport project. Thus, recalling the definition of $\rho(i,t,\mu)$, we can find \( \rho(i,t,\mu) + 1 \) is quite close to the congestion multiplier. The consistency between our theoretical result of $3/2$ and empirical value of $1.5$ implies that, for most mobility services or project appraisals involving a group of users, it is reasonable to assume that (1) the service process of transport systems can be approximated by a Poisson process and (2) the intensity of users' risk aversion satisfies \( R_2(t) \le 1 \), as formalized in Lemma \ref{benchmarkRRA} in next section, which implies users make decisions with consideration of no more than first two moments, i.e., mean and variance. Equivalently, these users treat time variability as a small risk in everyday travel decisions. They therefore favor routes with lower variability, but are willing to tolerate moderate variability when its elimination requires disproportionate cost.

However, for a specific user, higher-order risk preference may be important for various reasons, such as greater variability inherent in the mobility service under special events (e.g., disasters) or the importance of the purpose of travel (e.g., business travel). In such scenarios, users may exhibit \( R_2(t) > 1 \) and large $CV^2$ for users, which corresponds to the general result given in Eq. \eqref{generalQuadraRatio_EU}. This explains why premium service prices in practice can be several times higher than those of regular service, such as the price differentials between regular trains and high-speed trains, or between bus rapid transit service and regular buse service. 

\subsection{General results for users with differentiable utility functions}
In this section, we consider a general setting where users' utility functions are assumed only to be differentiable. That is, we relax the assumptions on both the functional form of the utility and service process. 

Proposition \ref{generalRatio} shows that the ratio $\rho \left( i,t,\mu \right)$ generally depends on coefficient of variation (CV), RRA coefficient, and RP coefficient. Note that RRA coefficient, and RP coefficient are defined by derivatives of utility functions of time as introduced in Section \ref{sec2.1}.



\begin{proposition}\label{generalRatio}
    Suppose users are risk-averse under the EU framework for the service instance $(t, \mu, \sigma)$, then
    \begin{equation}
        \rho \left( i,t,\mu \right) = \frac{R_2(\mu) CV^2}{2+R_2(\mu) R_3(\mu)CV^2} \label{generalRatio_EU} \notag
   \end{equation}
   where $R_2(t)$ and $R_3(t)$ are  coefficients of RRA and RP, respectively. 
\end{proposition}

\begin{proof}
    {Proof.} By rearranging terms and substituting for $\pi$ using the Lemma \ref{RP_local approxiamtion}, we have 
    $$
    \rho \left( i,t,\mu \right) =\frac{\sigma ^2}{2\mu}\frac{u''\left( \mu \right)}{u'\left( \mu \right)}\frac{1}{1-1/2\sigma ^2\frac{u^{\prime\prime\prime}\left( \mu \right)}{u'\left( \mu \right)}}\,\,=\frac{\sigma ^2}{2\mu ^2}\frac{1}{\frac{1}{\mu}\frac{u'\left( \mu \right)}{u''\left( \mu \right)}+1/2\sigma ^2\frac{1}{\mu ^2}\left( -\mu \frac{u^{\prime\prime\prime}\left( \mu \right)}{u''\left( \mu \right)} \right)}
    $$
    Since we have defined RRA coefficient $R_2(\mu) =\mu \frac{u''\left( \mu \right)}{u'\left( \mu \right)}$ and RP coefficient $R_3(\mu) =-\mu \frac{u^{\prime\prime\prime}\left( \mu \right)}{u''\left( \mu \right)}$, the ratio is 
    $$
    \rho \left( i,t,\mu \right) =\frac{1}{2}CV^2\frac{1}{1/R_2(\mu) +1/2CV^2R_3(\mu)}=\frac{R_2(\mu) CV^2}{2+R_2(\mu) R_3(\mu)CV^2}
    $$
    This completes the proof. $\square$
\end{proof}

For a service in a certain period, its $CV^2$ is an objective measure and thus can be viewed as a fixed value for all users. Consequently, the magnitude of the ratio for a specific service primarily depends on users' intensity of risk preferences. Importantly, benchmark values of the coefficients of relative risk aversion (RRA) and prudence (RP) can be identified to categorize the intensity of users’ risk preferences. These benchmark values serve as thresholds that distinguish how users trade off reductions associated with different moments of the time distribution. Particularly, reductions in mean, variance, and skewness  respectively reflect users’ preferences for shorter average times, lower dispersion of outcomes, and reduced exposure to extreme delays. Table~\ref{tab:moment_risk_preferences} summarizes the moment reductions in time variability and interpret them in terms of distributional change, user concern, and risk preference.

\vspace{-0.8em}
\begin{table}[ht]
\centering
\caption{Moment reductions in time variability and associated risk preferences}
\label{tab:moment_risk_preferences}
    \begin{tabularx}{\textwidth}{X>{\raggedright\arraybackslash}p{4.5cm}>{\raggedright\arraybackslash}p{3.6cm}X}
    \hline
    \textbf{Moment reduced} 
    & \textbf{Distributional change} 
    & \textbf{User concern} 
    & \textbf{Risk preference} \\
    \hline
    Mean 
    & Leftward shift of the time distribution 
    & Lower average  time 
    & Risk neutral \\
    Variance 
    & Contraction of the distribution around the mean 
    & Lower dispersion or volatility of outcomes 
    & Second-order risk preference (risk aversion) \\
    Skewness 
    & Reduction in the right-tail mass 
    & Lower exposure to extreme delays 
    & Third-order risk preference (prudence) \\
    \hline
    \end{tabularx}
\end{table}


\subsubsection{Benchmark value of RRA coefficient and its interpretation} \label{sec:Benchmark value of RRA} 
Lemma \ref{benchmarkRRA} provides
the benchmark value, i.e., $R_2(t) = 1$, to partition risk-averse users into different groups according to the degree of RRA. Its proof follows directly from \textcolor{blue}{\cite{eeckhoudt2009values}} and \textcolor{blue}{\cite{eeckhoudt2009values}}.

\vspace{0.5em}
\begin{lemma}[\textcolor{blue} {\cite{gollier2001economics}} and \textcolor{blue} {\cite{eeckhoudt2009values}}]\label{benchmarkRRA}
    For an EU risk-averse user, the benchmark value of RRA coefficient governing the trade-off between time reductions and variance reductions is $1$.
\end{lemma}
\vspace{0.5em}
In practice, such benchmark value of 1 serves as a natural threshold value for distinguishing how users trade off mean time savings against variance reductions. Based on Lemma \ref{benchmarkRRA}, two cases can be considered. When $R_2(t) > 1$, we have $t^2\left( -u''\left( t \right) \right) > t(-u'\left( t \right))$, indicating a higher intensity of risk aversion. In this case, the marginal benefit of reducing second moment  (i.e., variance) exceeds that of reducing the first moment (i.e., mean). Conversely, $R_2(t) < 1$ means $t^2\left( -u''\left( t \right) \right) < t(-u'\left( t \right))$, indicating that users prefer first moment reductions to second moment reductions. Indeed, the case of $R_2(t) \le 1$ corresponds exactly to the users with quadratic utility functions in Section \ref{Sec:SpecialQuadratic}. Based on these observations, Remark \ref{remarkQuadraRiskMoment} further elucidates the implications of quadratic utility function for risk-averse users in terms of risk and moments reduction preferences.

\begin{remark}[users' risk preferences under quadratic utility function]\label{remarkQuadraRiskMoment}
    Consider the vertex form of quadratic utility function: $u\left( t \right) =a\left( t-h \right) ^2+k$ where $h$ and $k$ are coordinates of the vertex. We assume $h\le 0$ and $k\le 0$. Based on the vertex form, the expected utility can be expressed as a function of mean and variance of time as follows:
    \begin{align}
        \mathbb{E}\left[ u\left( t \right) \right] & =\mathbb{E}\left[ a\left( t-h \right) ^2+k \right] \notag \\
        & =a\mathbb{E}\left\{ \left[ t-\mathbb{E}\left( t \right) \right] ^2 +2\left[ t-\mathbb{E}\left( t \right) \right] \left[ \mathbb{E}\left( t \right) -h \right] + \left[ \mathbb{E}\left( t \right) -h \right] ^2 \right\} + k \notag \\
        & = a\left\{ \text{Var}\left( t \right) +\left[ \mathbb{E}\left( t \right) -h \right] ^2 \right\} +k \label{EU_mean_variance}
    \end{align}
    
    If we ignore the term $k$ in the right-hand side of Eq. \eqref{EU_mean_variance}, we can clearly see that the changes of mean and variance affect $\mathbb{E}\left[ u\left( t \right) \right]$ values symmetrically, which corresponds to the case of $R_2(t) = 1$. However, as we all know, the shift of $k$ in Eq. \eqref{EU_mean_variance} can change the mean value $\mathbb{E}\left( t \right)$ but has no effects on the variance $\text{Var}\left( t \right)$. This asymmetry explains why we have $R_2(t) \le 1$ for general quadratic utility function, which in turn underlies the results in Proposition \ref{quadraticRatioExponetialTime} and Corollary \ref{quadraticRatioProp}. In other words, although users with quadratic utility functions dislike time variability, they tend to prioritize mean time savings over variability reductions.
\end{remark}

\subsubsection{Benchmark value of RP coefficient and its interpretation} Proposition \ref{benchmarkRP} establishes 2 as the benchmark value of RP coefficient and its proof is based on the idea of using multiplicative lotteries and the principle of harm dis-aggregation \textcolor{blue}{\citep{eeckhoudt2009values}}.

\begin{proposition}\label{benchmarkRP}
    For an EU risk-averse and prudent user, the benchmark value of RP coefficient governing the trade-off between reducing second moment and reducing third moment is $2$. 
\end{proposition}
\begin{proof}
    {Proof.} Consider a user facing the choice of two instances $S_1$ and $S_2$, as shown in Figure \ref{figureRP}, for a mobility service. In Figure \ref{figureRP}, $T_0$ is the initial time of this service with no time variability; $l$ represents a deterministic time and $\tilde{\gamma}$ is a zero-mean variability with $\tilde{\gamma}\in \left[ -1,+\infty \right)$. By construction, $S_1$ and $S_2$ have the same expected service time, while $S_1$ has lower variance but higher skewness than $S_2$.
    \begin{figure}[ht]
    \centering
    \vspace{-0.8em}
    \includegraphics[width=0.7\textwidth]{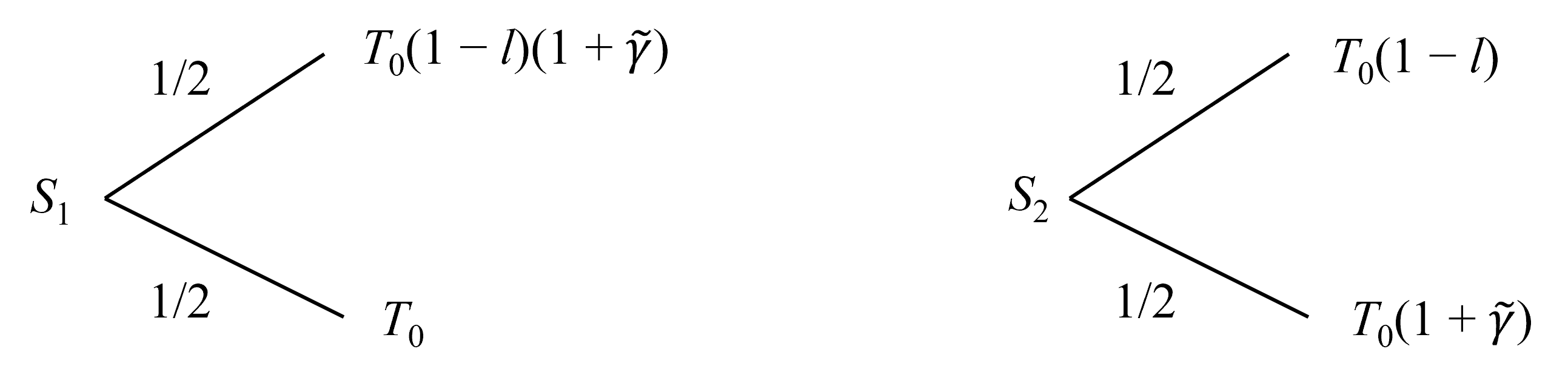}
    \vspace{-0.8em}
    \caption{Two instances of a mobility service with a certain loss and zero-mean risk.}
    \label{figureRP}
    \end{figure} 
    
    If one user prefers $S_2$ to $S_1$, then it means:
    $$
    \frac{1}{2}\mathbb{E}\left[ u\left( T_0\left( 1-l \right) \left( 1+\tilde{\gamma} \right) \right) \right] +\frac{1}{2}u\left( T_0 \right) \,\,\le \,\,\frac{1}{2}u\left( T_0\left( 1-l \right) \right) +\frac{1}{2}\mathbb{E}\left[ u\left( T_0\left( 1+\tilde{\gamma} \right) \right) \right] 
    $$
    Namely, we have 
    $$
    \mathbb{E}\left[ u\left( T_0\left( 1-l \right) \left( 1+\tilde{\gamma} \right) \right) \right] -u\left( T_0\left( 1-l \right) \right) \,\,\le \,\,\mathbb{E}\left[ u\left( T_0\left( 1+\tilde{\gamma} \right) \right) \right] -u\left( T_0 \right)  
    $$
    Denoting $\vartheta \left( l,\tilde{\gamma};T_0 \right) =\mathbb{E}\left[ u\left( T_0\left( 1-l \right) \left( 1+\tilde{\gamma} \right) \right) \right] -u\left( T_0\left( 1-l \right) \right)$, we have to show $\vartheta \left( l,\tilde{\gamma};T_0 \right) \le \vartheta \left( 0,\tilde{\gamma};T_0 \right)$ for all $l > 0$. So, its partial derivative in terms of $l$ should be non-positive, namely
    $$
    \frac{\partial \vartheta \left( l,\tilde{\gamma};T_0 \right)}{\partial l}=-T_0\mathbb{E}\left[ \left( 1+\tilde{\gamma} \right) u'\left( T_0\left( 1-l \right) \left( 1+\tilde{\gamma} \right) \right) \right] +T_0u'\left( T_0\left( 1-l \right) \right) \le 0
    $$
    By rearrangement, we have $\mathbb{E}\left[\left( 1+\tilde{\gamma} \right) u'\left( T_0\left( 1-l \right) \left( 1+\tilde{\varepsilon} \right) \right) \right] \ge u'\left( T_0\left( 1-k \right) \right)$. Now, let us define $\omega \left(l,\tilde{\gamma};T_0 \right) =\left( 1+\tilde{\gamma} \right) u'\left( T_0\left( 1 - l \right) \left( 1+\tilde{\gamma} \right) \right)$, then it implies $\mathbb{E}\left[ \omega \left( l,\tilde{\gamma};T_0 \right) \right] \ge \omega\left( l,\mathbb{E}\left[ \tilde{\gamma} \right] ;T_0 \right)$. In other words, $\omega \left(l,\tilde{\gamma};T_0 \right)$ should be a convex function with respect to $\tilde{\gamma}\in \left[ -1,+\infty \right)$.
    $$
    \frac{\partial ^2\omega \left( l,\tilde{\gamma};T_0 \right)}{\partial \tilde{\gamma}}=2T_0\left( 1-l \right) u''\left( t\left( 1-l \right) \left( 1+\tilde{\gamma} \right) \right) +\left( 1+\tilde{\gamma} \right) T_{0}^{2}\left( 1-l \right) ^2u^{\prime\prime\prime}\left( T_{0}^{2}\left( 1-l \right) \left( 1+\tilde{\gamma} \right) \right) \ge 0
    $$
    By rearrangement, we can further have the following inequality:
    $$
    2u''\left( T_0\left( 1-l \right) \left( 1+\tilde{\gamma} \right) \right) +\left( 1+\tilde{\gamma} \right) T_0\left( 1-l \right) u^{\prime\prime\prime}\left( T_0\left( 1-l \right) \left( 1+\tilde{\gamma} \right) \right) \ge 0
    $$
    Let $t = \left( 1+\tilde{\gamma} \right) T_0\left( 1-l \right)$, then we have $R_3\left( t \right) =-t\frac{u^{\prime\prime\prime}\left( t \right)}{u''\left( t \right)}\geq 2$. This completes the proof. $\square$
\end{proof}

Proposition \ref{benchmarkRP} and its proof show that, if RP coefficient $R_3(t) > 2$, for the same expected time, users prefer an alternative choice with a lower skewness, even it has higher variance. On the contrary, if $R_3(t) < 2$, users prefer an alternative choice with a higher variance even it has lower skewness. Consequently, $2$ serves as a benchmark value for the balance between reducing second moment  and reducing third moment. Although skewness-related behavior received limited attention in theoretical transport studies, it has been recognized by empirical studies, such as \textcolor{blue}{\cite{van2005monitoring}} and \textcolor{blue}{\cite{van2008travel}}. Furthermore, extensive studies have reported the right-skewed tail of travel time distributions in the literature \textcolor{blue}{\citep[e.g.,][]{van2005monitoring, susilawati2013distributions, zang2024value}}.




The benchmark values of RRA and RP provide a useful baseline for characterizing the intensity of different risk attitudes. In particular, empirically estimated ratios, such as those obtained from stated-preference surveys,  can be compared against these benchmark values to facilitate interpretation. As discussed in \textcolor{blue}{\cite{zang2022ttrb}}, a wide range of models requires users' risk parameters as inputs, including the preference parameters in VOR models, the probabilities needed for reliability measures, and the risk parameters in route choice models. However, specifying these values ex ante is challenging, and in many cases it is even impractical to gather each user's parameters through surveys, as users may not even be aware of their precise values and also these parameters may vary rather than remain fixed. The proposed benchmark provides a practical reference point, enabling the identifications of users' risk preferences through long-term and large-scale individual-level datasets collected through non-survey-based approaches, such as GPS and cellular data.

\subsubsection{The rate of return from reducing time variability}
With general differentiable utility functions, the reliability ratio can be expressed as follows:
$$
\eta(i, t, \mu) =\frac{1}{1+1/2R_2(\mu) R_3(\mu) CV^2}
$$
We can thus derive Corollary \ref{RRMarginalGeneral}, which shows that higher-order risk preference is associated with diminishing marginal benefits from reducing time variability. Specifically, although stronger higher-order risk preferences lead users to incur larger additional time costs to hedge against variability, the marginal benefit of further reductions in time variability becomes smaller. This result is intuitive: as mobility services represent the derived demand undertaken to achieve other objectives, consuming the service itself does not generate intrinsic utility. In addition, such a diminishing marginal effect is widely observed in other related scenarios. For example, in studies of the value of travel time and travel time reliability, \textcolor{blue}{\cite{metz2008myth}} found that the additional benefit derived from further travel time savings tends to diminish, while \textcolor{blue}{\cite{zang2024value}} found that the additional benefit derived from further improving travel time reliability also becomes diminishing. Furthermore, within the framework of reliable network design, \textcolor{blue}{\cite{xu2014modeling}} identified a diminishing marginal effect of improving network reliability performance relative to the construction budget.

\vspace{0.5em}
\begin{corollary}\label{RRMarginalGeneral}
    For risk-averse and prudent EU users, the marginal benefit of reducing time variability
    is at most $1/2$ times that of the marginal benefit of reducing time. Namely, $\eta(i, t, \mu) \le \frac{1}{2}$.
\end{corollary}
\vspace{0.5em}
\begin{proof}
    {Proof.} If users dislike longer time, then by definition we have $VOT \ge 0$. It follows from Eq. \eqref{VOTusingVOTT} that $VOT\left( \mu \right) -1/2\sigma ^2\frac{u^{\prime\prime\prime}\left( \mu \right)}{\varphi} \ge 0$. Namely, 
    \begin{equation*}
        -\frac{1}{\varphi}u'\left( \mu \right) \ge 1/2\sigma ^2\frac{u^{\prime\prime\prime}\left( \mu \right)}{\varphi}\ \Longrightarrow u'\left( \mu \right) \ge -1/2\sigma ^2u^{\prime\prime\prime}\left( \mu \right) \ \Longrightarrow 1 \le -\frac{1}{2}\sigma ^2\frac{u^{\prime\prime\prime}\left( \mu \right)}{u'\left( \mu \right)} 
    \end{equation*}
    where the second inequality and the third inequality are implied by $ \varphi <0 $ and $ u'\left( \mu \right) <0 $, respectively. With the definition of RRA coefficient and RP coefficient, we could have 
    \begin{equation*}     
        \frac{2}{\sigma ^2}\cdot \mu ^2 \le -\mu ^2\frac{u^{\prime\prime\prime}\left( \mu \right)}{u'\left( \mu \right)}=\left( -\mu \frac{u^{\prime\prime\prime}\left( \mu \right)}{u''\left( \mu \right)}\,\, \right) \cdot \left( \mu \frac{u''\left( \mu \right)}{u'\left( \mu \right)} \right) = R_2\left( \mu \right) R_3\left( \mu \right) 
    \end{equation*}
    That is, we have $2 \le CV^2 \cdot R_2\left( \mu \right) R_3\left( \mu \right)$ and therefore $\eta(i, t, \mu) \le \frac{1}{2}$. This completes the proof. $\square$
\end{proof}

\section{Extensions to non-expected utility models}
This section extends the proposed framework to non-expected utility models. We first examine the dual theory (DT; \textcolor{blue}{\cite{yaari1987dual}}), which captures users’ perceptions of probabilities, and then consider the more general rank-dependent utility (RDU; \textcolor{blue}{\cite{quiggin1982theory}}). Note that the RDU encompasses both EU and DT as special cases and forms the basis of prospect theory \textcolor{blue}{\citep{tversky1992advances}}.

Below we first introduce dual moments, which play a vital role in generalizing the results derived under EU models to non-EU model settings, as will be demonstrated in Sections \ref{sec:resultDU} and \ref{sec:resultRDU}.

\subsection{Dual moments}\label{sec:DualMoments}
For time $t$ with CDF $F(t)$, its first and second moments, $m_1$ and $m_2$, can be written by
\begin{equation}
    m_1 \coloneq \int{tdF\left( t \right)} \text{,\   }m_2 \coloneq \int{\left( t-m_1 \right) ^2dF\left( t \right)} \notag
\end{equation}
Following \textcolor{blue}{\cite{eeckhoudt2022dual}}, if we treat $m_2$ as a primal moment, the associated dual moment can be written as
\begin{equation}
    \overline{m}_2 \coloneq \int{\left( t-m_1 \right) d\left( F\left( t \right) \right) ^2} \label{secondDualMean}
\end{equation}

Therefore, a more precise terminology for $\overline{m}_2$ is \textit{the second dual moment about the mean}\endnote{\textcolor{blue}{\cite{eeckhoudt2022dual}} termed such dual moment as the \textit{maxiance} in the vein of the variance as $\overline{m}_2$ represents the expected best outcome obtained from two independent draws of random times.} of $t$. Note that $(F(t))^2$ is the CDF of maximum of independent and identically distributed random time. If we only consider $t^{(1)}$ and $t^{(2)}$ as two independent copies of $t$, the dual moment can be also written as 
\begin{equation}
    \overline{m}_2 \coloneq \mathbb{E}\left[ \max \left( t^{\left( 1 \right)},\ t^{\left( 2 \right)} \right) \right] -\mathbb{E}\left[ t \right] \label{secondDualMeanIID}
\end{equation}
Eq. \eqref{secondDualMeanIID} shows that $\overline{m}_2$ admits an order-statistic interpretation: it equals the expected value of the largest outcome from two independent draws of the risk. For more details of mean order statistics in the statistics literature, we refer the readers to \textcolor{blue}{\cite{david2004order}}. Similarly, the \textit{second dual moment about the variance} of $t$ should be computed as
\begin{equation}
    \overline{m}_{2}^{2} \coloneq \int{\left( t-m_1 \right) ^2d\left( F\left( t \right) \right) ^2} \label{secondDualVariance}
\end{equation}
Consider two independent copies $t^{(1)}$ and $t^{(2)}$ of $t$ again, we have $\overline{m}_{2}^{2} \coloneq \text{Var}\left[ \max \left( t^{\left( 1 \right)},\,\,t^{\left( 2 \right)} \right) \right]$.
Specifically, if we only consider the zero mean variability $\tilde{\xi}$ in random time, to be shown in next section, then we have 
$$
\overline{m}_{2}^{2}(\tilde{\xi})=\mathbb{E}\left[ \max \left( \left( \tilde{\xi}^{\left( 1 \right)} \right) ^2,\,\,\left( \tilde{\xi}^{\left( 2 \right)} \right) ^2 \right) \right] \ -\ \mathbb{E}\left[ \tilde{\xi} \right]
$$

\subsection{Results under dual theory}\label{sec:resultDU}
Under EU framework, random service time is characterized by its probability distribution, specified by the the PDF $f(t)$ and the CDF $F(t)$. In practice, however, users often perceive event probabilities only; they may overestimate some probabilities while underestimating others. To capture this behavior, dual theory (DT) replaces the CDF $F(t)$ with $w(F(t))$ where the probability weighting function $w$: $[0, 1]\longrightarrow[0, 1]$ satisfies the following conditions: $w\left( 0 \right) =0,\ w\left( 1 \right) =1,\ w'>0$. Then, the utility representation of a user for random time under DT is 
\begin{equation} \label{EU_DT}
    \mathbb{E}\left[ u_{\text{dual}}\left( t \right) \right] = \int{-tdw(F(t))},
\end{equation}
where strict concavity, i.e., $w^{\prime\prime} <0 $, implies aversion to mean-preserving spreads, which is analogous to the risk aversion in EU framework (i.e., $u^{\prime\prime} < 0$).

\subsubsection{Dual moments and variability premium}
Under DT, variability is introduced through probability (i.e.,  small probabilities attached to fixed payoffs), which is dual to the case of EU in which variability arises from random payoffs with equal probability. For illustration purpose, Figure \ref{figureEUDUVariability} presents illustrative examples of binary variability under EU and DT cases, highlighting their differences and underlying relationships for the same certain instance $(T_0)$.
\begin{figure}[ht]
    \centering
    \vspace{-0.8em}
    \includegraphics[width=0.90\textwidth]{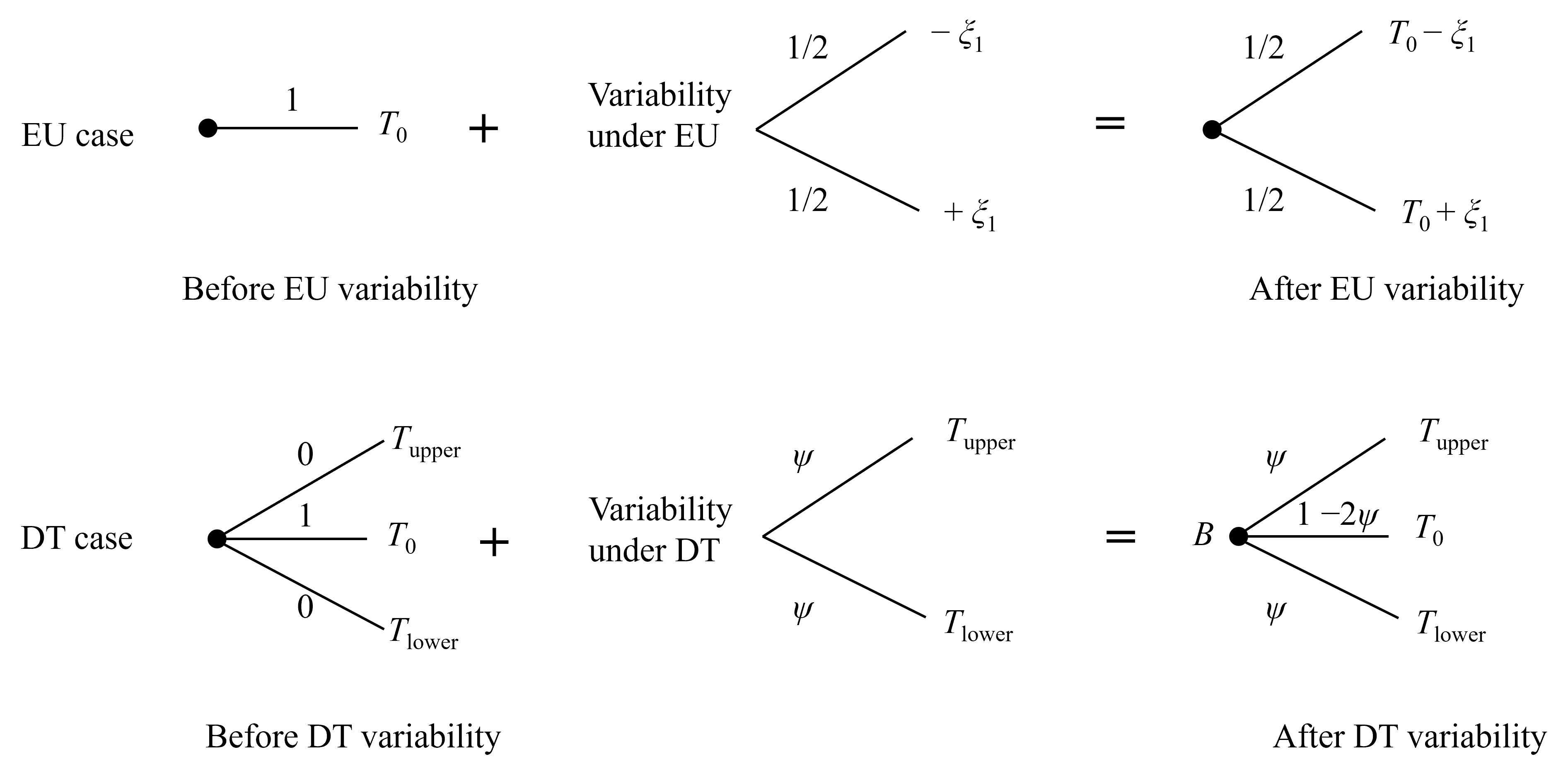}
    \vspace{-0.8em}
    \caption{Illustration of service instances with EU and DT variability where $\xi_1$ is a zero-mean variability and $\psi$ is probability.}
    \label{figureEUDUVariability}
\end{figure} 

Let us consider a more general instance$_{2\psi}(\mu)$, in which there are $n + 2$ possible outcomes. Specifically, we have $T_{\text{min}}$ with probability $p_0 - \psi$, $\mu$ with probability $2\psi$, $T_{\text{max}}$ with probability $1 - p_0 - \psi$, and all other outcomes with 0 as probability. Then, dual zero mean n-state variability $\tilde{\xi}$ is added into instance$_{2\psi}(\mu)$. Specifically, $n$ possible outcomes are in ascending order, i.e., $\xi _1\le \cdots \le \xi _i\le \cdots \le \xi _n$; $\xi _i$ is assumed to be $0$ and we assume each of $n$ possible outcomes occur with the equal probability of $\frac{2\psi}{n}$.The service change due to dual variability is illustrated in Figure \ref{figureServiceWithRiskDU}. As we allow adjacent states to be the same, such variability $\tilde{\xi}$ can be viewed as variability with unequal state probability as well. In addition, though we use discrete representations, it naturally extends to continuous representations if $n \longrightarrow \infty$. Specifically, if $2\psi = 1$, it will be completely the same as instances illustrated in Figure \ref{figureServiceWithRiskEU}; this is how we can generalize the following analysis to add variability to all probability. From the perspective of DT, it is conceptually more appropriate to formulate the analysis in terms of changes affecting only a portion of the probability distribution, rather than uniformly perturbing all probabilities.
\begin{figure}[ht]
    \centering
    \vspace{-0.8em}
    \includegraphics[width=0.9\textwidth]{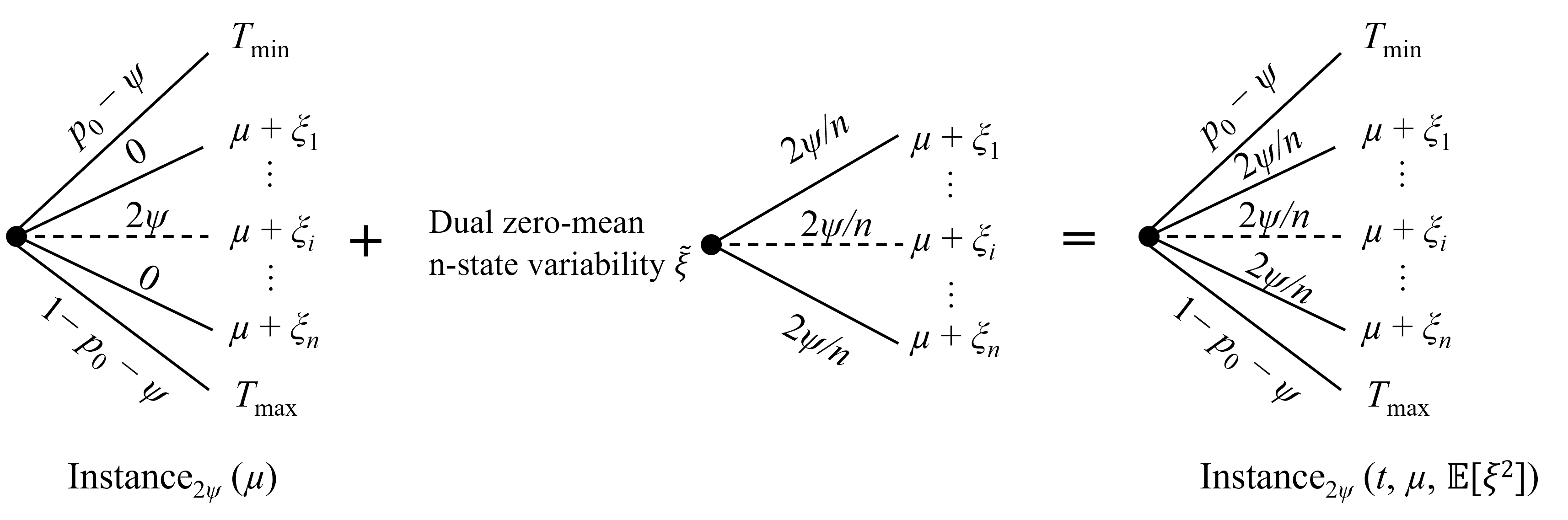}
    \vspace{-0.8em}
    \caption{Service instances before/after dual zero-mean $n$ state variability in which $\xi _1\le \cdots \le \xi _i=0\le \cdots \le \xi _n$.}
    \label{figureServiceWithRiskDU}
\end{figure} 

According to definitions and formulations of dual moments introduced in Section \ref{sec:DualMoments} (i.e., mainly Eq. \eqref{secondDualMean}), the dual moment $\overline{m}_2$ of variability $\tilde{\xi}$ can be given by
\begin{equation}
    \overline{m}_2=\sum_{i=1}^n{\xi _i}p\left( \xi _i \right) \left( 2\sum_{j=1}^i{p\left( \xi _j \right)}\right) =\frac{4\psi ^2}{n^2}\sum_{i=1}^n{\xi _i}\left( 2i\right) =\frac{4\psi ^2}{n^2}\sum_{i=1}^n{\xi _i}2i \label{eqDualmoment_mean}
\end{equation}

Then, as illustrated by Figure \ref{figureVPDU}, the variability premium under dual theory is defined such that users are \textit{indifferent} between risky instance$_{2\psi}(t, \mu, \mathbb{E}[\tilde{\xi}])$  and the corresponding certain instance$_{2\psi}(\mu + \pi_{\text{dual}})$. Proposition \ref{RP_DT} gives the formulation of $\pi_{\text{dual}}$ under DT. 

\begin{figure}[ht]
    \centering
    \vspace{-0.8em}
    \includegraphics[width=0.55\textwidth]{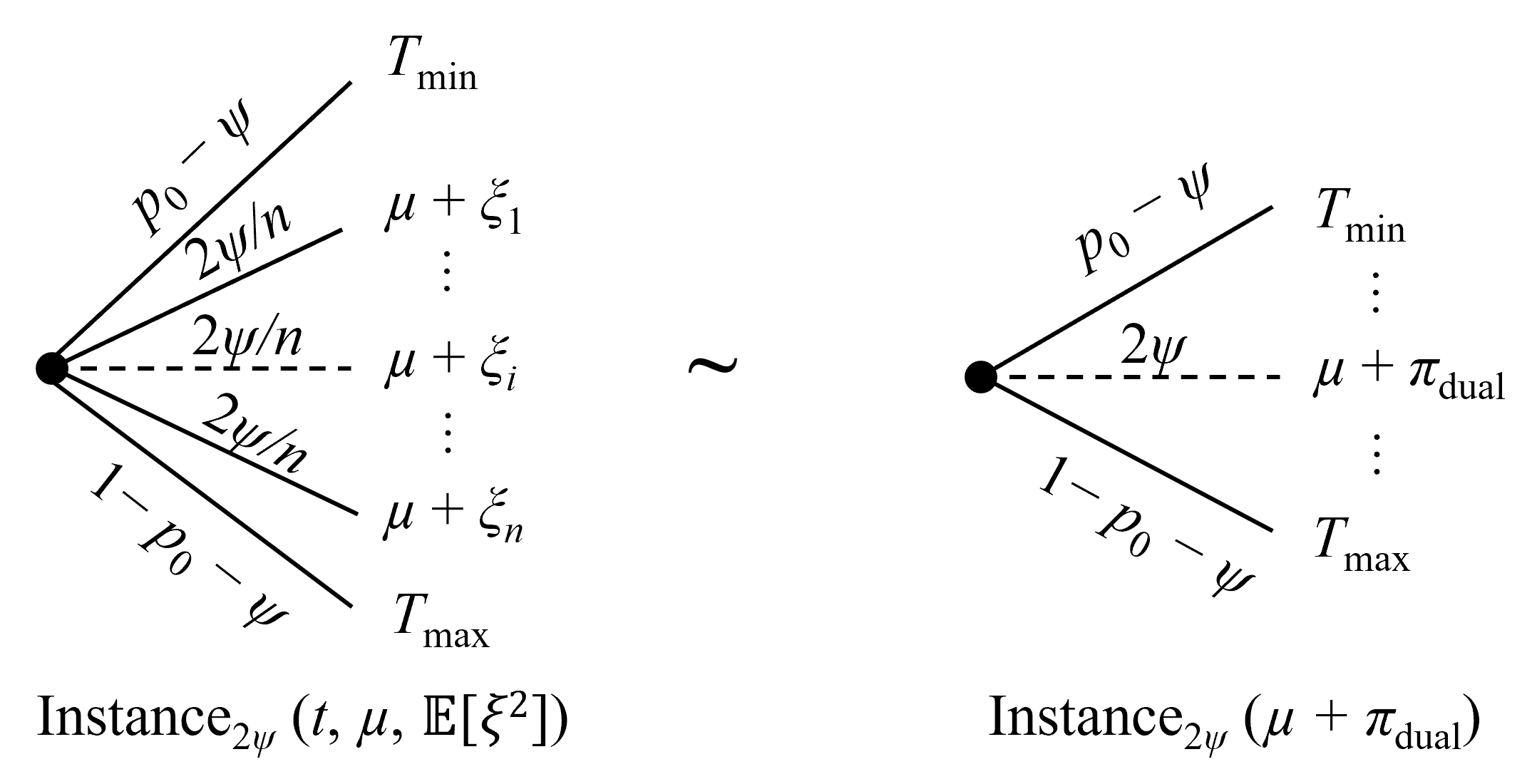}
    \vspace{-0.8em}
    \caption{Variability premium for DT users.}
    \label{figureVPDU}
\end{figure} 
\begin{proposition}\label{RP_DT}
    For a service with random service time $t$, the variability premium of a DT user with probability weighting function $w(F(t))$ can be approximated by
    \begin{equation} \label{RP_local approxiamtion_DU}
        \pi _{\text{dual}}\simeq \frac{1}{2}\cdot \frac{w''\left( p_0 \right)}{w'\left( p_0 \right)}\cdot \overline{m}_2 \notag
    \end{equation}
\end{proposition}
\begin{proof}
    {Proof.} If DT users are indifferent between risky instance$_{2\psi}(t, \mu, \mathbb{E}[\tilde{\xi}])$  and certain instance$_{2\psi}(\mu + \pi_{\text{dual}})$ as shown in Figure \ref{figureVPDU}, it indicates that
    \begin{equation} \label{RP_DT_indifferent}
        \left( w( p_0+\psi ) -w( p_0-\psi \right)) ( \mu +\pi _{\text{dual}}) =\sum_{i=1}^n{\left( w( p_0-\psi +i\cdot \frac{2\psi}{n}) -w( p_0-\psi +( i-1 ) \cdot \frac{2\psi}{n}) \right) \left( \mu +\xi _i \right)}
    \end{equation}
    By solving this equation, we can derive
    \begin{equation}\label{piDual}
    \pi _{\text{dual}}=\sum_{i=1}^n{\frac{w\left( p_0-\psi +i\cdot \frac{2\psi}{n} \right) -w\left( p_0-\psi +\left( i-1 \right) \cdot \frac{2\psi}{n} \right)}{w\left( p_0+\psi \right) -w\left( p_0-\psi \right)}\xi _i} \notag
    \end{equation}
    
    Applying second-order Taylor expansion to $w\left( p_0-\psi +i\cdot \frac{2\psi}{n} \right)$ and $w\left( p_0-\psi +\left( i-1 \right) \cdot \frac{2\psi}{n} \right)$ at $p_0$, we can obtain the following expression:
    $$
    w\left( p_0-\psi +i\cdot \frac{2\psi}{n} \right) -w\left( p_0-\psi +\left( i-1 \right) \cdot \frac{2\psi}{n} \right) =w'\left( p_0 \right) \frac{2\psi}{n}+\frac{1}{2}w''\left( p_0 \right) \left( -\frac{4\psi ^2}{n}+\frac{4\psi ^2\left( 2i-1 \right)}{n^2} \right)  
    $$
    Similarly, we have $w\left( p_0+\psi \right) -w\left( p_0-\psi \right) =2w'\left( p_0 \right) \psi$ by invoking second-order Taylor expansion at $p_0$. Consequently, $\pi _{\text{dual}}$ is  given by
    $$
    \pi _{\text{dual}}=\frac{1}{2}\cdot \frac{1}{2\psi}\frac{w''\left( p_0 \right)}{w'\left( p_0 \right)}\cdot \frac{4\psi ^2}{n^2}\sum_{i=1}^n{\xi _i2i}=\frac{1}{2}\cdot \frac{1}{2\psi}\frac{w''\left( p_0 \right)}{w'\left( p_0 \right)}\cdot \overline{m}_2
    $$ 
    where $\overline{m}_2$ is the dual moment about the mean of risk $\tilde{\xi}$. 
    Now, we consider adding dual variability to all possible times, and thus we have $2\psi = 1$. As a result, $\pi _{\text{dual}} = \frac{1}{2}\frac{w''\left( p_0 \right)}{w'\left( p_0 \right)}\overline{m}_2$. This completes the proof. $\square$
\end{proof}

As shown by Proposition \ref{variability premium}, the primal moment, variance, plays a fundamental measure of time variability under EU. In contrast, Proposition \ref{RP_DT} shows that the second dual moment plays an analogous role under the dual theory, standing on equal footing with variance as the fundamental measure of time variability \textcolor{blue}{\citep{eeckhoudt2022dual}}.


\subsubsection{COT, COTV, and the ratio under dual theory} As the utility function of DT is linear, so the key difference between $U(c, t)$ and $U_\text{dual}(c, t)$ is that $u(t)$ in $U(c, t)$ is replaced by $-t$ in  $U_\text{dual}(c, t)$. Thus, based on Eq. \eqref{EU_expression}, the expected utility of DT is given by
\begin{equation} \label{EU_U_dual}
    \mathbb{E}\left[ U_{\text{dual}}\left( c,\,\,t \right) \right] =-\varphi c+\int{-tdw( F\left( t \right))},
\end{equation}
and thus the VOT and COTV can be calculated by $ VOT=-\frac{1}{\varphi}$ and $COTV=\frac{-\mu +\int{tdw\left( F\left( t \right) \right)}}{\varphi}$.

Based on the above results, Proposition \ref{DualRatio} gives the formulation of the ratio of the COTV to COT.

\begin{proposition}\label{DualRatio}
    For DT risk-averse users, the ratio depends on CV of random times and dual risk aversion parameters only:
    \begin{equation} \label{generalRatio_DU}
        \rho _{\text{dual}}\left( i,t,\mu \right) =\frac{1}{2}CV_{\text{dual}}^{2}\left( \mu \frac{w''\left( p_0 \right)}{w'\left( p_0 \right)} \right) \notag  
    \end{equation}
\end{proposition}
\begin{proof}
    {Proof.} Recalling n-state variability in Figure \ref{figureVPDU} and considering $n\longrightarrow \infty$, we have
    \begin{equation}\label{COTVDualIntegeral}
        \int{-tdw\left( F\left( t \right) \right)}= -\sum_{i=1}^n{\left( w\left( p_0-\psi +i\cdot \frac{2\psi}{n} \right) -w\left( p_0-\psi +\left( i-1 \right) \cdot \frac{2\psi}{n} \right) \right) \left( \mu +\xi _i \right)}
    \end{equation}

    Based on Eq. \eqref{RP_DT_indifferent}, we have 
    $$\int{tdw\left( F\left( t \right) \right)}=\left( \mu +\pi _{\text{dual}} \right) \left( w\left( p_0+\psi \right) -w\left( p_0-\psi \right) \right) =\mu +\pi _{\text{dual}}$$
    where $\psi = 1/2$ and $p_0 =1/2$ for Eq. \eqref{RP_DT_indifferent}. Finally, the COTV is given by 
    \begin{equation}\label{COTVDual}
        COTV=-\frac{\pi _{\text{dual}}}{\varphi}=\pi _{\text{dual}}\cdot VOT
    \end{equation}
    According to the definition of the ratio $\rho _{\text{dual}}$, we have the final result. This completes the proof. $\square$
\end{proof}

Interestingly, the result of $\rho _{\text{dual}}\left( i,t,\mu \right)$ is quite similar to that under EU framework with quadratic utility function in Eq. \eqref{generalQuadraRatio_EU} of Section \ref{Sec:SpecialQuadratic}. Namely, under DT framework, the ratio $\rho _{\text{dual}}\left( i,t,\mu \right)$ also depends on the degree of time variability and users' second-order preference of risk. Specifically, for a DT risk-averse user, the second dual moment is as significant as the variance in measuring COTV for mobility services. Based on the modeling rationality of EU and DT, we can find that the variance quantifies time variability in terms of “payoff plane” (i.e., time), the second dual moment represents a measure of time variability in terms of “probability plane” (i.e., perceived probability). With respect to quantifying users' preference, EU measures it through second-order risk preference embodied in utility function, while DT characterize them through second-order “weighting” of perception about probability.

\subsection{Results under rank dependent model}\label{sec:resultRDU}
Although the dual theory (DT) is able to address the common ratio effect, it lacks coherence. A more general and behaviorally realistic framework is the rank dependent utility (RDU) model. The key distinction between RDU and DT lies in the treatment of utility: RDU replaces the linear utility in DT with a general utility function. Accordingly, the expected utility under RDU is written as
\begin{equation} \label{EU_RDU}
    \mathbb{E}\left[ u_{\text{RDU}}\left( t \right) \right] = \int{u(t)dw(F(t))}
\end{equation}
Proposition \ref{RP_RDU} gives the formulation of $\pi_{\text{rank}}$ under RDU, which makes users indifferent between risky instance$_{2\psi}(t, \mu, \mathbb{E}[\tilde{\xi}])$  and the corresponding certain instance$_{2\psi}(\mu + \pi_{\text{rank}})$. 

\begin{proposition}\label{RP_RDU}
    For a mobility service with random service time $t$, the variability premium of an RDU user with probability weighting function $w(F(t))$ can be approximated by
    \begin{equation} \label{RP_local approxiamtion_RDU}
        \pi _{\text{rank}}\simeq \,\frac{1}{2}\,\frac{u''\left( \mu \right)}{u'\left( \mu \right)}m_2+\frac{1}{2}\,\frac{w''\left( p_0 \right)}{w'\left( p_0 \right)}\,\overline{m}_2+\frac{1}{2}\frac{u''\left( \mu \right)}{u'\left( \mu \right)}\cdot \frac{1}{2}\,\frac{w''\left( p_0 \right)}{w'\left( p_0 \right)}\left( \overline{m}_{2}^{2}-m_2 \right) 
    \end{equation}
\end{proposition}
\begin{proof}
    {Proof.} If RDU users are indifferent between the risky instance$_{2\psi}(t, \mu, \mathbb{E}[\tilde{\xi}])$ (as in Figure \ref{figureVPDU}) and a newly corresponding certain instance$_{2\psi}(\mu + \pi_{\text{rank}})$, it means
    \begin{align}
        & \left( w\left( p_0+\psi \right) -w\left( p_0-\psi \right) \right) u\left( \mu +\pi _{\text{rank}} \right) = \notag \\
        &\ \ \ \ \ \ \ \ \ \ \ \ \sum_{i=1}^n{\left( w\left( p_0-\psi +i\cdot \frac{2\psi}{n} \right) -w\left( p_0-\psi +\left( i-1 \right) \cdot \frac{2\psi}{n} \right) \right) u\left( \mu +\xi _i \right)} \label{indifferentRDU}
    \end{align}

    Following \textcolor{blue}{\cite{eeckhoudt2022dual}}, we apply second-order Taylor expansion to $u\left( \mu +\xi \right)$ at $\mu$, $w\left( p_0-\psi +i\cdot \frac{2\psi}{n} \right)$ and $w\left( p_0-\psi +\left( i-1 \right) \cdot \frac{2\psi}{n} \right)$ at $\mu$, $w\left( p_0+\psi \right)$ and $w\left( p_0-\psi \right)$ at $p_0$, and first-order Taylor expansion to  $u\left( \mu +\pi _{\text{rank}} \right)$ at $\mu$. Then, we can obtain the solution to the above equation as
    \begin{align}
        2w'\left( p_0 \right) \psi & \left( u\left( \mu \right) +u'\left( \mu \right) \pi _{\text{rank}} \right) = \notag\\
        &\sum_{i=1}^n{\left( w'\left( p_0 \right) \frac{2\psi}{n}+\frac{1}{2}w''\left( p_0 \right) \left( -\frac{4\psi ^2}{n}+\frac{4\psi ^2\left( 2i-1 \right)}{n^2} \right) \right) \left( u\left( \mu \right) +u'\left( \mu \right) \xi _i+\frac{u''\left( \mu \right)}{2}\xi _{i}^{2} \right)} \notag
    \end{align}
    After simplification (see Appendix for the detailed derivations), we can have:
    $$
    \pi _{\text{rank}}=\,\frac{m_2}{2}\,\frac{u''\left( \mu \right)}{u'\left( \mu \right)}+\frac{\overline{m}_2}{2}\,\frac{1}{2\psi}\frac{w''\left( p_0 \right)}{w'\left( p_0 \right)}\,+\frac{1}{2}\cdot \frac{1}{2}\cdot \frac{1}{2\psi}\,\left( \frac{w''\left( p_0 \right)}{w'\left( p_0 \right)} \right) \left( \frac{u''\left( \mu \right)}{u'\left( \mu \right)} \right) \left( \overline{m}_{2}^{2}-m_2 \right)
    $$
    where $m_2$, $\overline{m}_{2}$ and $\overline{m}_{2}^{2}$ are primal second moment, second dual moment about mean, and second dual moment about variance defined in Section \ref{sec:DualMoments}. Now, we consider adding variability to all possible times, and thus we have $2\psi = 1$. Then, the resulting expression is given in Eq. \eqref{RP_local approxiamtion_RDU}.
    $\square$
\end{proof}

Note that $COT$ and $COTV$ under RDU can be estimated by following the same process used to estimate $VOT$ and $COTV$ under EU in Section \ref{sec:COTandCOTV_EU}. Consequently, we have Proposition \ref{RDURatio}.
\begin{proposition}\label{RDURatio}
    For RDU risk-averse users, the ratio depends on CV of random times and dual risk aversion parameters only:
    \begin{equation} \label{generalRatio_RDU}
        \rho \left( i,t,\mu \right) =\tau _h\left( \frac{1}{2}\,\frac{m_2}{\mu ^2}\,\mu \frac{u''\left( \mu \right)}{u'\left( \mu \right)}+\frac{1}{2}\frac{\overline{m}_2}{\mu ^2}R_{2}^{\text{dual}}\left( p_0 \right) \,+\frac{1}{2}\left( \frac{w''\left( p_0 \right)}{w'\left( p_0 \right)} \right) \cdot \frac{1}{2}R_2\left( \mu \right) \frac{\overline{m}_{2}^{2}-m_2}{\mu ^2} \right)  \notag  
    \end{equation}
    where $\tau _h$ is the parameter of the $VOT(\mu)$ to the $VOT(\mu_h)$.
\end{proposition}
\begin{proof}
    {Proof.} For random times $t$ under RDU framework, we let $t=\mu _w+\sigma _wy$ where $\mu_w$ and $\sigma_w$ are mean and variance under $w(F(t))$. Then, according to the proof of Proposition \ref{rho_local approxiamtion}, we have
    \begin{align} \label{VOT_COTV_rank}
        & VOT\approx g_w\left( 0 \right) +\sum\limits_{l=1}^2{\frac{g_{w}^{\left( l \right)}\left( 0 \right)}{l!}\sigma ^l}=VOT\left( \mu _w \right) -1/2\sigma _{w}^{2}\frac{u^{'''}\left( \mu _w \right)}{\varphi} \notag\\
        & COTV=\frac{u\left( \mu \right) -\int{u\left( t \right) dw\left( F\left( t \right) \right)}}{\varphi} \notag
    \end{align}

    Note that $\int{u\left( t \right) dw\left( F\left( t \right) \right)}$ is right-hand side of Eq. \eqref{indifferentRDU}. By invoking first order Taylor expansion, the COTV is further written as
    $$
    COTV=\frac{u\left( \mu \right) -\int{u\left( t \right) dw\left( F\left( t \right) \right)}}{\varphi} \approx \frac{u\left( \mu \right) -u\left( \pi _{\text{rank}}+\mu \right)}{\varphi}=\pi _{\text{rank}}VOT\left( \mu \right) 
    $$
    Let $\tau _h$ be the parameter of the $VOT(\mu)$ to the $VOT(\mu_h)$. The ratio of $\rho_{rank} \left( i,t,\mu \right)$ under RDU can be reformulated as
    \begin{align*}
        \rho \left( i,t,\mu \right) &=\frac{\pi _{\text{rank}}VOT\left( \mu \right)}{\mu VOT\left( \mu _h \right)} \\
        & =\tau _h\frac{1}{\mu}\,\frac{m_2}{2}\,\frac{u''\left( \mu \right)}{u'\left( \mu \right)}+\frac{\overline{m}_2}{2}\,\frac{1}{2\varepsilon}\frac{w''\left( p_0 \right)}{w'\left( p_0 \right)}\,+\frac{1}{2}\cdot \frac{1}{2}\cdot \frac{1}{2\varepsilon}\,\left( \frac{w''\left( p_0 \right)}{w'\left( p_0 \right)} \right) \left( \frac{u''\left( \mu \right)}{u'\left( \mu \right)} \right) \left( \overline{m}_{2}^{2}-m_2 \right) \\
        & =\tau _h\left( \frac{1}{2}\,\frac{m_2}{\mu ^2}R_2\left( \mu \right) \,+\frac{1}{2}\frac{\overline{m}_2}{\mu ^2}R_{2}^{\text{dual}}\left( p_0 \right) \,+\frac{1}{2}\left( \frac{w''\left( p_0 \right)}{w'\left( p_0 \right)} \right) \cdot \frac{1}{2}R_2\left( \mu \right) \frac{\overline{m}_{2}^{2}-m_2}{\mu ^2} \right) 
    \end{align*}
   This completes the proof. $\square$
\end{proof}

Proposition \ref{RDURatio} demonstrates that, for a RDU risk-averse user, the result of $\rho _{\text{rank}}\left( i,t,\mu \right)$ is also structurally similar to that under EU framework. In particular, it depends on the degree of time variability and users’ preferences. Specifically, both the first and second primal moments and the first and second dual moments jointly determine the COTV for mobility services. Accordingly, both users’ preference in terms of utility function and “weighting” function of perception about probability are important.

We can thus observe the relationship among the results of the ratio $\rho = COVT/COT$ between EU, DT and RDU frameworks according to Propositions \ref{generalRatio}, \ref{DualRatio} and \ref{RDURatio}. For RDU users characterized by Eq. \eqref{EU_RDU}, if we use $F(t)$ to replace $w(F(t))$, then RDU would reduce to EU model. Accordingly, the result in Proposition \ref{RDURatio} reduces to Proposition \ref{generalRatio}. Alternatively, if we use $t$ to replace $u(t)$ in Eq. \eqref{EU_RDU}, then RDU would reduce to DT model. Accordingly, the result in Proposition \ref{RDURatio} reduces to Proposition \ref{DualRatio}. In short, in either DT or EU model, the last term in Proposition \ref{RDURatio} under RDU framework, which captures the intersection of DT and EU in measuring time variability and users' preference, vanishes.

\section{Conclusion}
This paper investigates the economic consequences of time variability in mobility services, with the aim of clarifying how users are affected by uncertainty in service times. Our analysis shows that time variability does not affect users arbitrarily: its economic impact follows a structured and predictable pattern jointly shaped  by service time variability and behavioral attitudes toward risk. 

In a commonly adopted special case—quadratic utility combined with a Poisson service process—we showed that the total user cost under time variability does not exceed $3/2$ of the deterministic benchmark. More generally, without assumptions on the service process, variability costs scale with $\frac{1}{2}CV^2$ relative to time costs. These results indicate that substantial welfare losses arise primarily when variability is large relative to mean service time, whereas moderate variability leads to proportionate and predictable impacts. Under quadratic utility, we further found that the marginal benefit of reducing variability is equal to that of reducing mean time. This finding provides theoretical support for the claim commonly used in the literature \textcolor{blue}{\cite[see e.g.,][]{chen2023conservative, zang2024relia}} highlighting the importance of considering the COTV: users value time variability reductions as much as, if not more than, they value time savings. 

Beyond this special case, we show that the ratio of COTV to COT depends on the coefficient of variation (CV) of  times and on users’ higher-order risk preferences, captured by relative risk aversion (RRA) and relative prudence (RP). Benchmark values of these parameters provide a useful basis for distinguishing how users trade off reductions in mean, variance, and skewness. When sensitivity to higher-order risk becomes pronounced, diminishing marginal effects emerge: for risk-averse and prudent users, the marginal benefit of reducing time variability is bounded by $1/2$ of the marginal benefit of reducing mean time. Finally, extending the analysis to non–EU models, including DT and RDU, shows that these insights are robust across alternative behavioral representations. Although variability and preferences are modeled differently, the structural dependence of variability costs on relative dispersion and risk attitudes remains intact.

From a practical standpoint, characterizing the maximum economic loss associated with time variability offers a useful complement to existing valuation methods, particularly in settings where detailed data are limited and decisions must be made at an early stage. Moreover, the derived bounds provide a principled upper limit on users’ willingness to pay for reliability improvements, offering guidance for investment prioritization and the pricing and design of reliability-oriented mobility services.

\section*{Appendix}
To solve the solution in Proposition \ref{RP_RDU}, we have the following equation:
\begin{equation}
\begin{split}
    u\left( \mu \right) + & u'\left( \mu \right) \pi _{\text{rank}} = \frac{1}{2w'\left( p_0 \right) \varphi}\,\sum_{i=1}^n{\frac{2\varphi}{n}w'\left( p_0 \right) \left( u\left( \mu \right) +u'\left( \mu \right) \xi _i+\frac{u''\left( \mu \right)}{2}\xi _{i}^{2} \right)} \\
    & +\frac{1}{2w'\left( p_0 \right) \varphi}\,\sum_{i=1}^n{\frac{1}{2}w''\left( p_0 \right) \left( -\frac{4\varphi ^2}{n}+\frac{4\varphi ^2\left( 2i-1 \right)}{n^2} \right) \left( u\left( \mu \right) +u'\left( \mu \right) \xi _i+\frac{u''\left( \mu \right)}{2}\xi _{i}^{2} \right)} \notag
\end{split}
\end{equation}

The first term of right hand side of the above formulation can be simplified as
\begin{equation}
    \frac{1}{2w'\left( p_0 \right) \varphi}\,\sum_{i=1}^n{\frac{2\varphi}{n}w'\left( p_0 \right) \left( u\left( \mu \right) +u'\left( \mu \right) \xi _i+\frac{u''\left( \mu \right)}{2}\xi _{i}^{2} \right)}=u\left( \mu \right) \,\,+0+\sum_{i=1}^n{\frac{1}{n}\frac{u''\left( \mu \right)}{2}\xi_i^2} \notag
\end{equation}
Then, we have 
\begin{equation}
    u'\left( \mu \right) \pi _{\text{rank}}=\,\,\frac{u''\left( \mu \right)}{2n}\sum_{i=1}^n{\xi _{i}^{2}}+\,\frac{1}{2w'\left( p_0 \right) \varphi}\,\sum_{i=1}^n{\frac{1}{2}w''\left( p_0 \right) \left( -\frac{4\varphi ^2}{n}+\frac{4\varphi ^2\left( 2i-1 \right)}{n^2} \right) \left( u\left( \mu \right) +u'\left( \mu \right) \xi _i+\frac{u''\left( \mu \right)}{2}\xi _{i}^{2} \right)} \notag
\end{equation}
Note that the mean of $\xi$ is $0$ and we have $\sum_{i=1}^n{-\frac{4\varphi ^2}{n}+\frac{4\varphi ^2\left( 2i-1 \right)}{n^2}}=0$  and $\sum_{i=1}^n{\frac{4\varphi ^2\left( 2i-1 \right)}{n^2}\xi _i}=\overline{m}_2$, we have 
\begin{equation}
    \pi _{\text{rank}}=\,\frac{1}{2}\,\frac{u''\left( \mu \right)}{u'\left( \mu \right)}m_2 + \frac{\overline{m}_2}{2}\,\frac{1}{2\varphi}\frac{w''\left( p_0 \right)}{w'\left( p_0 \right)}\,+\frac{1}{2}\cdot \frac{1}{2}\,\frac{1}{2\varphi}\left( \frac{w''\left( p_0 \right)}{w'\left( p_0 \right)} \right) \left( \frac{u''\left( \mu \right)}{u'\left( \mu \right)} \right) \,\sum_{i=1}^n{\left( -\frac{4\varphi ^2}{n}+\frac{4\varphi ^2\left( 2i-1 \right)}{n^2} \right) \xi _{i}^{2}} \notag
\end{equation}

Referring to Eq. \eqref{eqDualmoment_mean}, we have 
\begin{equation}
    \overline{m}_{2}^{2}=\sum_{i=1}^n{\xi _i^2}p\left( \xi _i \right) \left( 2\sum_{j=1}^i{p\left( \xi _j \right)} \right) =\frac{4\psi ^2}{n^2}\sum_{i=1}^n{\xi _{i}^{2}}2i \notag
\end{equation}

Finally, we have
\begin{equation*}
    \pi _{\text{rank}}=\,\frac{1}{2}\,\frac{u''\left( \mu \right)}{u'\left( \mu \right)}m_2+\frac{\overline{m}_2}{2}\,\frac{1}{2\varphi}\frac{w''\left( p_0 \right)}{w'\left( p_0 \right)}\,+\frac{1}{2}\cdot \frac{1}{2}\,\frac{1}{2\varphi}\left( \frac{w''\left( p_0 \right)}{w'\left( p_0 \right)} \right) \left( \frac{u''\left( \mu \right)}{u'\left( \mu \right)} \right) \,\left( -m_2+\overline{m}_{2}^{2} \right)
\end{equation*}

\section*{Acknowledgement}
This research is supported by the National Research Foundation, Singapore under its AI Singapore Pro-
gramme (AISG Award No:AISG3-RP-2022-031).

\begingroup \parindent 0pt \parskip 0.0ex \def\enotesize{\normalsize} \theendnotes \endgroup

\bibliographystyle{informs2014}
\bibliography{reference}

\end{document}